\documentclass[11pt, a4paper]{article} 

\usepackage[dvipsnames]{xcolor}
\usepackage{amsfonts}
\usepackage{amssymb}
\usepackage{amsmath}
\usepackage{amsthm}
\usepackage{latexsym}
\usepackage{graphicx}
\usepackage{epstopdf}
\usepackage[retainorgcmds]{IEEEtrantools}
\usepackage{pdfpages}
\usepackage{enumerate}
\usepackage{caption}
\usepackage{amsthm}
\usepackage{subfig}
\usepackage{setspace}
\usepackage{dsfont}
\usepackage{bbm}
\usepackage{hyperref}
\usepackage{datetime}
\usepackage[colorinlistoftodos, textsize=footnotesize, backgroundcolor=GreenYellow!25, linecolor=GreenYellow!25, bordercolor=GreenYellow!60]{todonotes}
\presetkeys{todonotes}{fancyline}{}

\makeatletter

\newdimen\mainfontsize \mainfontsize=1\@ptsize pt
\topskip=\mainfontsize
   \@maxdepth=\maxdepth

\makeatother

\makeatletter
\renewcommand\@makefntext[1]{%
  \noindent\makebox[0.5em][r]{\@makefnmark}#1}
\makeatother

\makeatletter
\def\blfootnote{\xdef\@thefnmark{}\@footnotetext}
\makeatother

\newtheorem{theorem}{Theorem}[section]
\newtheorem{corollary}[theorem]{Corollary}
\newtheorem{lemma}[theorem]{Lemma}
\newtheorem{proposition}[theorem]{Proposition}

\theoremstyle{remark}

\theoremstyle{definition}

\newtheorem{remark}[theorem]{Remark}
\newtheorem{assumption}[theorem]{Assumption}

\newcommand{\ud}{\,\mathrm{d}}
\newcommand{\vd}{\mathrm{d}}
\newcommand{\R}{\mathbb{R}}
\newcommand{\F}{\mathcal{F}}
\newcommand{\E}{\mathbb E}
\newcommand{\N}{\mathbb N}

\newcommand{\lt}{\left}
\newcommand{\rt}{\right}
\newcommand{\pt}{\partial}

\def\P{{\mathbb P}}

\newcommand{\Ind}{\mathbbm{1}}
\newcommand{\supp}{\mathop{\mathrm{supp}}\nolimits}
\newcommand{\Var}{\mathop{\mathrm{Var}}\nolimits}

\makeatletter
\@addtoreset{equation}{section}

\makeatother

\title{Asset liquidation under drift uncertainty and regime-switching volatility}       
\author{Juozas Vaicenavicius\footnote{Department of Information Technology, Uppsala University, Box 337, 751 05 Uppsala, Sweden. 
\mbox{Email: \href{mailto:juozas.vaicenavicius@it.uu.se}{\nolinkurl{juozas.vaicenavicius@it.uu.se} }} 
} 
}
\date{Tuesday 31$^{\mathrm{st}}$ July, 2018}

\begin{document}

%

\maketitle
\begin{abstract}
Optimal liquidation of an asset with unknown constant drift and stochastic regime-switching volatility is studied. The uncertainty about the drift is represented by an arbitrary probability distribution; the stochastic volatility is modelled by $m$-state Markov chain. Using filtering theory, an equivalent reformulation of the original problem as a four-dimensional optimal stopping problem is found and then analysed by constructing approximating sequences of three-dimensional optimal stopping problems. 
An optimal liquidation strategy and various structural properties of the problem are determined. Analysis of the two-point prior case is presented in detail, building on which, an outline of the extension to the general prior case is given.


\end{abstract}


\section{Introduction}

Selling is a fundamental and ubiquitous economic operation. As the prices of goods fluctuate over time, `What is the best time to sell an asset to maximise revenue?'  qualifies as a basic question in Finance. Suppose that an asset needs to be sold before a known deterministic time $T>0$ and that the only source of information available to the seller is the price history. A natural mathematical reformulation of the aforementioned optimal selling question is to find a selling time $\tau^{*} \in \mathcal{T}_{T}$ such that
\begin{equation} \label{E:sell}
	\E[S_{\tau^*}] = \sup_{\tau \in \mathcal{T}_{T}} \E[S_\tau],
\end{equation}
where $\{S_{t}\}_{t\geq0}$ denotes the price process and $\mathcal{T}_{T}$ denotes the set of stopping times with respect to the price process $S$. 

Many popular continuous models for the price process are of the form 
\begin{IEEEeqnarray}{rCl} \label{E:Stock1}
\ud S_t = \alpha S_t \ud t +\sigma(t) S_t \ud W_t,
\end{IEEEeqnarray} 
where $\alpha \in \R$ is called the \emph{drift}, and $\sigma \geq 0$ is known as the \emph{volatility process}. Imposing simplifying assumptions that the volatility is independent of $W$ as well as time-homogeneous, an $m$-state time-homogeneous Markov chain stands out as a basic though still rather flexible stochastic volatility model (proposed in \cite{MKR}), which we choose to use in this article.  The flexibility comes from the fact that we can choose the state space as well as the transition intensities between the states. 

Though the problem \eqref{E:sell} in which $S$ follows \eqref{E:Stock1} 
is well-posed mathematically, from a financial point of view, the known drift assumption is widely accepted to be unreasonable (e.g.~see \cite[Section 4.2 on p.~144]{R13}) and needs to be relaxed. Hence, using the Bayesian paradigm, we model the initial uncertainty about the drift by a probability distribution (known as the \emph{prior} in Bayesian inference), which incorporates all the available information about the parameter and its uncertainty (see \cite{EV16} for more on the interpretation of the prior). If the quantification of initial uncertainty is subjective, then the prior represents one's beliefs about how likely the drift is to take different values. To be able to incorporate arbitrary prior beliefs, we set out to solve the optimal selling problem \eqref{E:sell} under an arbitrary prior for the drift.

In the present paper, we analyse and solve the asset liquidation problem \eqref{E:sell} in the case when $S$ follows \eqref{E:Stock1} with $m$-state time-homogeneous Markov chain volatility and unknown drift, the uncertainty of which is modelled by an arbitrary probability distribution. The first time a particular four-dimensional process hits a specific boundary determining the stopping set is shown to be optimal. This stopping boundary has attractive monotonicity properties and can be found using the approximation procedure developed.   

Let us elucidate our study of the optimal selling problem in more depth. Using the nonlinear filtering theory, the original selling problem with parameter uncertainty is rewritten as an equivalent optimal stopping problem of a standard form (i.e.~without unknown parameters). In this new optimal stopping problem, the posterior mean serves as the underlying process and acts as a stochastic creation rate; the payoff function in the problem is constant. The posterior mean is shown to be the solution of an SDE depending on the prior and the whole volatility history. Embedding of the optimal stopping problem into a Markovian framework is non-trivial because the whole posterior distribution needs to be included as a variable. Fortunately, we show that having fixed the prior, the posterior is fully characterised by only two real-valued parameters: the posterior mean and, what we call, the \emph{effective learning time}. As a result, we are able to define an associated Markovian value function with four underlying variables (time, posterior mean, effective learning time, and volatility) and study the optimal stopping problem as a four-dimensional Markovian optimal stopping problem (the volatility takes values in a finite set, but slightly abusing terminology, we still call it a dimension). Exploiting that the volatility is constant between the regime switches, we construct $m$ sequences of simpler auxiliary three-dimensional Markovian optimal stopping problems whose values in the limit converge monotonically to the true value function. The main advantage of this approximating sequence approach comparing with tackling the full variational inequality of the problem directly is that dealing with the analytically complicated coupled system is avoided altogether. Instead only much simpler standard uncoupled free-boundary problems need to be analysed or solved numerically to arrive at a desired result. We show that the value function is decreasing in time and effective learning time as well as increasing and convex in posterior mean. The first hitting time of a region specified by a stopping boundary that is a function of time, effective learning time, and volatility is shown to be optimal. The stopping boundary is increasing in time, effective learning time, and is the limit of a monotonically increasing sequence of boundaries from the auxiliary problems. Moreover, the approximation procedure using the auxiliary problems yields a method to calculate the value function as well as the optimal stopping boundary numerically. 

In the two-point prior case, the posterior mean fully characterises the posterior distribution, making the problem more tractable and allowing us to obtain some additional results. In particular, we prove that, under a skip-free volatility assumption, the Markovian value function is decreasing in the volatility and that the stopping boundary is increasing in the volatility. 

In a broader mathematical context, the selling problem investigated appears to be the first optimal stopping problem with parameter uncertainty and stochastic volatility to be studied in the literature. Thus it is plausible that ideas presented herein will find uses in other optimal stopping problems of the same type; for example, in classical problems of Bayesian sequential analysis (e.g.~see \cite[Chapter VI]{PS06}) with stochastically evolving noise magnitude. It is clear to the author that with additional efforts a number of results of the article can be refined or generalised. However, the objective chosen is to provide an intuitive understanding of the problem and the solution while still maintaining readability and clarity. This also explains why, for the most part, we focus on the two-point prior case and outline an extension to the general prior case only at the end.

\subsection{Related literature}

There is a strand of research on asset liquidation problems in models with regime-switching volatility, alas, they either concern only a special class of suboptimal strategies or treat the drift as observable. In \cite{qZ01}, a restrictive asset liquidation problem was proposed and studied; the drift as well as the volatility were treated as unobservable and the possibility to learn about the parameters from the observations was disregarded. The subsequent papers \cite{qZ02}, \cite{qZ06}, \cite{qZ08} explored various aspects of the same formulation. An optimal selling problem with the payoff $e^{-r\tau}(S_{\tau}-K)$ was studied in \cite{bO07} for the Black-Scholes model, in \cite{GZ05} for a two-state regime-switching model, and in \cite{qZ06} for an $m$-state model with finite horizon. In all three cases, the drift and the volatility are assumed to be fully observable.

In another strand of research, the optimal stopping problem \eqref{E:sell} has been solved and analysed in the Black-Scholes model under arbitrary uncertainty about the drift. The two-point prior case was studied in \cite{EL11}, while the general prior case was solved in \cite{EV16} using a different approach. This article can be viewed as a generalisation of \cite{EV16} to include stochastic regime-switching volatility.  Related option valuation problems under incomplete information were studied in \cite{G}, \cite{mV17}, both in the two-point prior case, and in \cite{DMV} in the $n$-point prior case.

The approach we take to approximate a Markovian value function by a sequence of value functions of simpler constant volatility problems was used before in \cite{LW} to investigate a finite-horizon American put problem (also, its slight generalisation) in a regime-switching model with full information. Regrettably, in the case of $3$ or more volatility states, the recursive approximation step in \cite[Section 5]{LW} contains a blunder; we rectify it in Section \ref{ss:approx_seq} of this article. A possible alternative route to analysing and solving the optimal stopping problem is to  analytically tackle the system of variational inequalities directly using weak solutions techniques (e.g., see \cite{BL, hP}), similarly as in  \cite{C11} for American options with regime-switching volatility. Structural and regularity properties would need to be established using PDE techniques. If appropriate theoretical results can be obtained, numerical PDE schemes discussed in \cite{H11} should yield a numerical solution. However, this alternative approach requires a different toolkit, appears to be more demanding analytically, and hence not investigated further in the present article. 

Though it is true that the current paper is a generalisation of \cite{EV16} from constant volatility to the regime-switching stochastic volatility model, the extension is definitely not a straightforward one. Novel statistical learning intuitions were needed, and new proofs were developed to arrive at the results of the paper. 
One of the main insights of the optimal liquidation problem with constant volatility in \cite{EV16} was that the current time and price were sufficient statistics for the optimal selling problem. However, changing the volatility from constant to stochastic makes the posterior distribution of the drift truly dependent on the price path. This raises questions whether an optimal liquidation problem can be treated using the mainstream finite-dimensional Markovian techniques at all, and also whether any of the developments from the constant volatility case can be taken advantage of. In the two-point prior case with regime-switching volatility, the following new insight was key. Despite the posterior being a path-dependent function of the stock price, we can show that the current time, posterior mean and instantaneous volatility (extracted from the price process) are sufficient statistics for the optimal liquidation problem. Alas, for any prior with more than two points in the support, the same triplet is no longer a sufficient statistic. Fortunately, if in addition to the time-price-volatility triplet we introduce an additional statistic, which we name the effective learning time, the resulting $4$-tuple becomes a sufficient statistic for the selling problem under a general prior. Besides these insights, some new technicalities (in particular, Lemma \eqref{L:Indep}) stemming from stochastic volatility had to be resolved to reformulate the optimal selling problem into the standard Markovian form. 

In relation to \cite{LW}, though we employ the same general iterative approximation idea to construct an approximating sequence for the Markovian value function, the particulars, including proofs and results, are notably distinct. Firstly, we work in a more general setting, proving and formulating more abstract as well as, in multiple instances, new type of results. For example, we prove things in the $m$-state rather than the two-state regime-switching model. This allowed us to catch and correct an erroneous construction of the approximating sequence in \cite{LW} for models with more than two volatility states. Moreover, almost all the proofs follow different arguments either because of the structural differences in the selling problem or because we prefer another way, which seems to be more transparent and direct, to arrive at the results. Lastly, many of the results in the present paper are problem-specific and even not depend on the iterative approximation of the value function after all.

The idea to iteratively construct a sequence of auxiliary value functions that converge to the true value function in the limit is generic and has been many times successfully applied to optimal stopping problems with a countable number of discrete events (e.g.~jumps, discrete observations). In the setting with partial observations, an iterative approximation scheme was employed in \cite{BDK06} to study the Poisson disorder detection problem with unknown post-disorder intensity, then later, in \cite{DPS08}, to analyse a combined Poisson-Wiener disorder detection problem,  and, more recently, in \cite{BK15}, to investigate the Wiener disorder detection under discrete observations.
In the fully observable setting, such iterative approximations go back to at least as early as \cite{G86}, which deals with a Markovian optimal stopping problem with a piecewise deterministic underlying. In Financial Mathematics, iteratively constructed approximations were used in \cite{eB09} and \cite{eB11} to study the value functions of finite and perpetual American put options, respectively, for a jump diffusion. Besides optimal stopping, the iterative approximation technique was utilised for the singular control problem \cite{EL14} of optimal dividend policy. 

\section{Problem set-up} \label{S:PF}

We model a financial market on a filtered probability space $(\Omega,\F, \{\mathcal{F}_{t} \}_{t \geq 0}, \P)$ satisfying the usual conditions. Here the measure $\P$ denotes the physical probability measure. The price process is modelled by 

\begin{IEEEeqnarray}{rCl} \label{E:Stock}
\ud S_t = X S_t \ud t +\sigma(t) S_t \ud W_t,
\end{IEEEeqnarray} 
where $X$ is a random variable having probability distribution $\mu$, $W$ is a standard Brownian motion,  and $\sigma$ is a time-homogeneous right-continuous $m$-state Markov chain with a generator $\Lambda = (\lambda_{ij})_{1\leq i,j \leq m}$ and taking values $\sigma_{m} \geq \ldots \geq \sigma_{1}  >0$. Moreover, we assume that $X$, $W$, and $\sigma$ are independent. Since the volatility can be estimated from the observations of $S$ in an arbitrary short period of time (at least in theory), it is reasonable to assume that the volatility process $\{ \sigma(t) \}_{t \geq 0}$ is observable. Hence the available information is modelled by the filtration $\mathbb F^{S, \sigma} = \lt\{ \mathcal F^{S, \sigma}_t \rt\}_{t \geq 0}$ generated by the processes $S$ and $\sigma$ and augmented by the null sets of $\F$. Note that the drift $X$ and the random driver $W$ are not directly observable.

 
The optimal selling problem that we are interested in is
\begin{IEEEeqnarray}{rCl} \label{E:OS}
V=\sup_{\tau \in \mathcal{T}_T^{S, \sigma} } \E[S_\tau], 
\end{IEEEeqnarray}
where $\mathcal{T}_T^{S, \sigma}$ denotes the set of  $\mathbb F^{S, \sigma}$-stopping times that are smaller or equal to a prespecified time horizon $T >0$. 

\begin{remark}
It is straightforward to include a discount factor $e^{-r\tau}$ in \eqref{E:OS}. In fact, it simply corresponds to a shift of the prior distribution $\mu$ in the negative direction by $r$.  
\end{remark}

Let $l:=\inf \supp (\mu)$ and $h:=\sup \supp (\mu)$. It is easy to see that if $l \geq0$, then it is optimal to stop at the terminal time $T$. 
Likewise, if $h \leq 0$, then stopping immediately, i.e.~at time zero, is optimal. The rest of the article focuses on the remaining and most interesting case.
\begin{assumption} 
$l < 0 < h$.
\end{assumption}

\subsection{Equivalent reformulation under a measure change}
Let us write $\hat X_{t} := \E [X \,|\, \mathcal{F}^{S, \sigma}_{t} ]$. Then the process
\[ \hat W_t := \int_0^t \frac{1}{\sigma(s)}(X - \hat X_s)  \ud s + W_t ,\]
called the innovation process, is an $\mathbb{F}^{S,\sigma}$-Brownian motion (see \cite[Proposition 2.30 on p.~33]{BC09}). 

\begin{lemma} \label{L:Indep}
The volatility process $\sigma$ and the innovation process $\hat W$ are independent.
\end{lemma}
\begin{proof}
Since $X$, $W$, and $\sigma$ are independent, we can think of $(\Omega, \F, \P)$ as a product space $ \lt( \Omega_{X,W} \times \Omega_{\sigma}, \F_{X,W} \otimes \F_{\sigma}, \P_{X,W} \times \P_{\sigma} \rt)$. Let 
$
A, A' \in \mathcal{B}(\R^{[0,T]}) 
$.
Then
\begin{IEEEeqnarray*}{rCl}
\P \lt( \hat W \in A, \, \sigma \in A' \rt) &=& \int_{\Omega_{X,W} \times \Omega_{\sigma}} \Ind_{\{ \hat W(\omega_{X,W}, \omega_{\sigma}) \in A, \, \sigma(\omega_{\sigma}) \in A' \}} \ud \lt( \P_{X,W}\times \P_{\sigma} \rt)(\omega_{X,W}, \omega_{\sigma}) \\
&=& \int_{\Omega_{\sigma}} \int_{\Omega_{X, W}} 
\Ind_{\{ \hat W(\omega_{X,W}, \omega_{\sigma}) \in A\}} \Ind_{\{\sigma(\omega_{\sigma}) \in A' \}} \ud \P_{X,W}(\omega_{X,W}) \ud \P_{\sigma}(\omega_{\sigma}) \\
&=& \int_{\Omega_{\sigma}}  \Ind_{\{\sigma(\omega_{\sigma}) \in A' \}} \int_{\Omega_{X, W}} 
\Ind_{\{ \hat W(\omega_{X,W}, \omega_{\sigma}) \in A\}} \ud \P_{X,W}(\omega_{X,W}) \ud \P_{\sigma}(\omega_{\sigma}) \\
&=& \int_{\Omega_{\sigma}} \Ind_{\{ \sigma(\omega_{\sigma}) \in A' \}} \P_{X,W} \lt( \hat W(\cdot, \omega_{\sigma}) \in A\rt) \ud \P_{\sigma}(\omega_{\sigma}) \\
&=& \P \lt( \hat W \in A\rt) \P_{\sigma} \lt( \sigma \in A' \rt) \\
&=& \P \lt( \hat W \in A\rt) \P \lt( \sigma \in A' \rt), \label{E:IndSplit} \IEEEyesnumber
\end{IEEEeqnarray*}
where the penultimate equality is justified by the fact that, for any fixed $\omega_{\sigma}$, the innovation process $\hat W (\cdot, \omega_{\sigma})$ is a Brownian motion under $\P_{X,W}$. Hence from \eqref{E:IndSplit}, the processes $\hat W$ and $\sigma$ are independent.
\end{proof}

Defining a new equivalent measure $\tilde{\P} \sim \P$ on $(\Omega, \F_{T})$ via the Radon-Nikodym derivative
\[ 
\frac{\vd \tilde{\P}}{\vd \P}  = e^{\int_{0}^{T} \sigma(t) \ud \hat{W}_t - \frac{1}{2}\int_{0}^{T}\sigma(t)^2\ud t}
\]
and writing
\begin{IEEEeqnarray*}{rCl}
S_t &=& S_0e^{Xt+ \int_{0}^{t} \sigma(s) \ud W_s-\frac{1}{2} \int_{0}^{t}\sigma(s)^{2} \ud s}\\
&=& S_0e^{\int_0^t\hat X_s \ud s + \int_{0}^{t} \sigma(s) \ud \hat W_s-\frac{1}{2}\int_{0}^{t}{\sigma(s)^2}\ud s},
\end{IEEEeqnarray*}
we have that, for any  $\tau \in \mathcal{T}^{S, \sigma}_{T}$,
\begin{IEEEeqnarray*}{rCl}
 \E \lt[ S_\tau \rt] 
= \tilde{\E} \lt[ S_0e^{ \int_0^\tau  \hat X_{s} \ud s }  \rt]= S_0 \tilde{\E} \lt[e^{ \int_0^\tau  \hat X_{s} \ud s } \rt].
\end{IEEEeqnarray*}
Moreover, by Girsanov's theorem, the process $B_t:= -\int_{0}^{t} \sigma(s) \ud s + \hat{W}_t$ is a $\tilde{\P}$-Brownian motion on $[0, T]$. In addition, Lemma \ref{L:Indep} together with \cite[Proposition 3.13]{BC09} tells us that the law of $\sigma$ is the same under $\tilde{\P}$ and $\P$, as well as that $B$ and $\sigma$ are independent under $\tilde{\P}$.

Without loss of generality, we set $S_0=1$ throughout the article, so the optimal stopping problem \eqref{E:OS} can be cast as
\begin{IEEEeqnarray}{rCl} \label{E:OSN}
V=\sup_{\tau \in \mathcal{T}^{S, \sigma}_T } \tilde{\E}[e^{ \int_0^\tau  \hat X_{s} \ud s }].
\end{IEEEeqnarray}
Between the volatility jumps, the stock price is a geometric Brownian motion with known constant volatility and unknown drift. Hence, by Corollary 3.4 in \cite{EV16}, we have that $\mathbb F^{S, \sigma} = \mathbb F^{\hat X, \sigma}$ and $\mathcal{T}^{S, \sigma}_T=\mathcal{T}^{\hat X, \sigma}_T$, where $\mathbb F^{\hat X, \sigma}$ denotes the usual augmentation of the filtration generated by $\hat X$ and $\sigma$,     also, $\mathcal{T}^{\hat X , \sigma}_{T}$ denotes the set of $\mathbb F^{\hat X, \sigma}$-stopping times not exceeding $T$. As a result, an equivalent reformulation of \eqref{E:OSN} is
\begin{IEEEeqnarray}{rCl} \label{E:OSN2}
V=\sup_{\tau \in \mathcal{T}^{\hat X, \sigma}_T } \tilde{\E}[e^{ \int_0^\tau  \hat X_{s} \ud s }],
\end{IEEEeqnarray}
which we will study in the subsequent parts of the article.

\subsection{Markovian embedding}

In all except the last section of this article, we will focus on the special case when $X$ has a two-point distribution $\mu = \pi \delta_{h} + (1-\pi) \delta_{l}$, where $h>l$, $\pi \in (0,1)$ are constants, and $\delta_{h}, \delta_{l}$ are Dirac measures at $h$ and $l$, respectively. In this special case, expressions are simpler and arguments are easier to follow than in the general prior case; still, most underlying ideas of the arguments are the same. Hence, we choose to understand the two-point prior case first, after which generalising the results to the general prior case will become a rather easy task.

Since the volatility is a known constant between the jump times, using the dynamics of $\hat X$ in the constant volatility case (the equation (3.9) in \cite{EV16}), the process $\hat X$ is a unique strong solution of  
\begin{IEEEeqnarray}{rCl} \label{E:Xhat}
\vd \hat X_{t} &=& \sigma(t) \phi(\hat X_{t}, \sigma(t)) \ud t + \phi(\hat X_{t}, \sigma(t)) \ud B_{t},
\end{IEEEeqnarray} 
where
\begin{IEEEeqnarray*}{rCl}
\phi(x, \sigma) &:=&  \frac{1}{\sigma}(h-x)(x-l). 
\end{IEEEeqnarray*}

Now, we can embed the optimal stopping problem \eqref{E:OSN} into a Markovian framework by defining a Markovian value function
\begin{IEEEeqnarray}{rCl} \label{E:OSM}
v(t,x, \sigma) := \sup_{\tau \in \mathcal T_{T-t} } \tilde{\E}[e^{ \int_0^\tau  \hat X^{t,x, \sigma}_{s} \ud s }], \quad (t,x, \sigma) \in [0,T] \times (l, h) \times \{ \sigma_{1},\ldots, \sigma_{m} \}.
\end{IEEEeqnarray}
Here $\hat X^{t,x, \sigma}$ denotes the process $\hat X$ in \eqref{E:Xhat} started at time $t$ with $\hat X_{t} = x$, $\sigma(t) = \sigma$, and $\mathcal T_{T-t}$ stands for the set of stopping times less or equal to $T-t$ with respect to the usual augmentation of the filtration generated by $\{ \hat X^{t,x,\sigma}_{t+s}\}_{s \geq 0}$ and $\{ \sigma(t+s)\}_{s \geq 0}$.
The formulation \eqref{E:OSM} has an interpretation of an optimal stopping problem with the constant payoff $1$ and the discount rate $-\hat X_{s}$; from now onwards, we will study this discounted problem. The notation $v_{i} := v(
\cdot, \cdot, \sigma_{i})$ will often be used.
\section{Approximation procedure} \label{S:Appr}
It is not clear how to compute $v$ in \eqref{E:OSM} or analyse it directly. Hence, in this section, we develop a way to approximate the value function $v$ by a sequence of value functions, corresponding to simpler constant volatility optimal stopping problems. 

\subsection{Operator $J_{i}$} \label{S:Ji}
For the succinctness of notation, let $\lambda_{i}:=\sum_{j \neq i} \lambda_{ij}$ denote the total intensity with which the volatility jumps from state $\sigma_{i}$. 
Also, let us define
\begin{IEEEeqnarray*}{rCl}
\eta^{t}_{i} &:=& \inf \{ s > 0 \,|\, \sigma(t+s) \neq \sigma(t) = \sigma_{i}\}, 
\end{IEEEeqnarray*}
which is an Exp($\lambda_{i}$)-distributed random variable representing the duration up to the first volatility change if started from the volatility state $\sigma_{i}$ at time $t$. 

Furthermore, let us define an operator $J$ acting on a bounded  $f : [0, T] \times (l,h) \to \R$ by 
\begin{IEEEeqnarray}{rCl}
\IEEEeqnarraymulticol{3}{l}{
(J f)(t, x, \sigma_{i}) } \nonumber \\
&&:= \sup_{\tau \in \mathcal{T}_{T-t}} \tilde{\E} \lt[ e^{ \int_0^\tau  \hat X^{t,x, \sigma_{i}}_{t+s}  \ud s } \Ind_{\{\tau < \eta^{t}_{i} \}} + e^{ \int_0^{\eta^{t}_{i}}  \hat X^{t,x,\sigma_{i}}_{t+s}  \ud s }f(t+\eta^{t}_{i}, \hat X^{t,x, \sigma_{i}}_{t+\eta^{t}_{i}}) \Ind_{\{\tau \geq \eta^{t}_{i}\}} \rt]  \label{E:J_i}\\
&&= \sup_{\tau \in \mathcal{T}_{T-t}} \tilde{\E} \lt[ e^{ \int_0^\tau  \hat X^{t,x, \sigma_{i}}_{t+s} -\lambda_{i} \ud s } + \lambda_{i} \int_{0}^{\tau} e^{ \int_0^u  \hat X^{t,x, \sigma_{i}}_{t+s} -\lambda_{i} \ud s }f(t+u, \hat X^{t,x,\sigma_{i}}_{t+u}) \ud u \rt], \label{E:J_iD}
\end{IEEEeqnarray}
where $\mathcal T_{T-t}$ denotes the set of stopping times less or equal to $T-t$ with respect to the usual augmentation of the filtration generated by $\{ \hat X^{t,x,\sigma_{i}}_{t+s}\}_{s \geq 0}$ and $\{ \sigma(t+s)\}_{s \geq 0}$. To simplify notation, we also define an operator $J_{i}$ by 
\begin{IEEEeqnarray*}{rCl} 
J_{i} f := (Jf)(\cdot, \cdot, \sigma_{i}).
\end{IEEEeqnarray*}
Intuitively, $(J_{i} f)$ represents a Markovian value function corresponding to optimal stopping before $t+ \eta^{t}_{i}$, i.e.~before the first volatility change after $t$, when, at time $t+\eta^{t}_{i} < T$, the payoff $f\lt(t+\eta^{t}_{i}, \hat X^{{t,x, \sigma_{i}}}_{t+\eta^{t}_{i}} \rt)$ is received provided stopping has not occurred yet.

\begin{proposition} \label{T:Ji}
Let $f : [0, T] \times (l,h) \to \R$ be bounded. Then 
\begin{enumerate}[(i)]
\item
$J f$ is bounded;
\item
$f$ increasing in the second variable $x$ implies that $J f$ is increasing in the second variable $x$;
\item \label{T:decrt}
$f$ decreasing in the first variable $t$ implies that $Jf$ is decreasing in the first variable $t$;
\item  \label{T:ic}
$f$ increasing and convex in the second variable $x$ implies that $Jf$ is increasing and convex in the second variable $x$;
\item
$J$ preserves order, i.e. $f_{1} \leq f_{2}$ implies $J f_{1} \leq J f_{2}$;
\item $J f \geq 1$.
\end{enumerate}
\end{proposition}
\begin{proof}
All except claim (\ref{T:ic}) are straightforward consequences of the representation \eqref{E:J_iD}. To prove (\ref{T:ic}), we will approximate the optimal stopping problem \eqref{E:J_iD} by Bermudan options. 

Let $i$ and $n$ be fixed. We will approximate the value function  $J_{i}f$ by a value function $w^{(f)}_{i,n}$ of a corresponding Bermudan problem with stopping allowed only at times $\lt\{ \frac{kT}{2^{n}} \,:\, k\in \{0,1, \ldots, 2^{n}\}\rt\}$.  We define $w^{(f)}_{i,n}$ recursively as follows. First,
\begin{IEEEeqnarray*}{rCl}
w^{(f)}_{i,n}(T,x) := 1.
\end{IEEEeqnarray*}
Then, starting with $k =2^{n}$ and continuing recursively down to $k=1$, we define
\begin{IEEEeqnarray}{rCl} \label{E:wStep}
w^{(f)}_{i,n}(t,x) &=& \lt\{ \begin{array}{ll}
g(t,x, \frac{kT}{2^{n}}), & t \in ( \frac{(k-1)T}{2^{n}}, \frac{kT}{2^{n}}), \\ 
g(\frac{(k-1)T}{2^{n}},x, \frac{kT}{2^{n}}) \vee 1, & t = \frac{(k-1)T}{2^{n}}, \\ 
\end{array} \rt.
\end{IEEEeqnarray}
where the function $g$ is given by
\begin{IEEEeqnarray}{rCl} \label{E:gB}
g(t,x, \frac{kT}{2^{n}}) &:=& \tilde{\E} \bigg[ e^{ \int_t^{\frac{kT}{2^{n}}}  \hat X^{t,x, \sigma_{i}}_{s} -\lambda_{i} \ud s } w^{(f)}_{i,n}\lt(\frac{kT}{n},\hat X^{t,x, \sigma_{i}}_{\frac{kT}{2^{n}}} \rt) \nonumber\\
 &&+ \int_{t}^{\frac{kT}{2^{n}}} e^{ \int_t^u  \hat X^{t,x, \sigma_{i}}_{s} -\lambda_{i} \ud s }f(u, \hat X^{t,x,\sigma_{i}}_{u}) \ud u \bigg].
\end{IEEEeqnarray}

Next, we show by backward induction on $k$ that $w^{(f)}_{i,n}$ is increasing and convex in the second variable $x$. Suppose that for some $k \in \{1, 2,\ldots, 2^{n}\}$, the function  $w^{(f)}_{i,n}\lt(\frac{kT}{2^{n}}, \cdot \rt)$ is increasing and convex (the assumption clearly holds for the base step $k=2^{n}$). Let $t\in [ \frac{(k-1)T}{2^{n}}, \frac{kT}{2^{n}})$. Then, since $f$ is also increasing and convex in the second variable $x$, we have that the function $g(t,\cdot, \frac{kT}{2^{n}})$, and so $w^{(f)}_{i,n}(t,\cdot)$, is convex by \cite[Theorem 5.1]{ET08}. 
 Moreover, from \eqref{E:gB} and \cite[Theorem~IX.3.7]{RY}, it is clear that $w^{(f)}_{i,n}(t,\cdot)$ is increasing. Consequently, by backward induction, we obtain that the Bermudan value function $w^{(f)}_{i,n}$ is increasing and convex in the second variable.

Letting $n \nearrow \infty $, the Bermudan value $w^{(f)}_{i,n} \nearrow J_{i}f$ pointwise. As a result, $J_{i}f$ is increasing and convex in the second argument, since convexity and monotonicity are preserved when taking pointwise limits.
\end{proof}
The sets
\begin{IEEEeqnarray*}{rCl}
 \mathcal{C}^{f}_{i} &:=& \{ (t,x) \in [0,T)\times (l, h) \,:\, (J_{i}f) (t,x) > 1 \}, \IEEEyesnumber \label{E:CntS}\\
  \mathcal{D}^{f}_{i} &:=& \{ (t,x) \in [0,T]\times (l, h) \,:\, (J_{i}f) (t,x) = 1 \} = [0,T]\times (l, h) \setminus \mathcal{C}^{f}_{i},
\end{IEEEeqnarray*}
correspond to continuation and stopping sets for the stopping problem $J_{i}f$ as the next proposition shows.
\begin{proposition}[Optimal stopping time] \label{T:OSTJ}
The stopping time  
\begin{IEEEeqnarray}{rCl} \label{E:Sfi}
\tau^{f}_{\sigma_{i}}(t,x) &=& \inf \{ u \in [0, T-t] \,:\, (t + u, \hat X^{t,x, \sigma_{i}}_{t+u}) \in \mathcal{D}^{f}_{i} \} 
\end{IEEEeqnarray} 
is optimal for the problem \eqref{E:J_iD}.
\end{proposition} 
\begin{proof}
A standard application of Theorem D.12 in \cite{KS2}.
\end{proof}

\begin{proposition} \label{T:reg}
If a bounded $f :[0,T] \times (l,h)  \to \R$ is decreasing in the first variable as well as increasing and convex in the second, then $J_{i}f$ is continuous.
\end{proposition}
\begin{proof}The argument is a trouble-free extension of the proof of the third part of Theorem 3.10 in \cite{EV16}; still, we include it for completeness.  
Before we begin, in order to simplify notation, we will write $u:= J_{i} f$.

Firstly, we let $r\in (l, h)$ and will prove that there exists $K>0$ such that, for every $t\in [0,T]$, the map $x\mapsto J_{i}f(t,x)$ is $K$-Lipschitz continuous on $(l, r]$. To obtain a contradiction, assume that there is no such $K$. Then, by convexity of $u$ in the second variable,  there is a sequence $\{ t_{n}\}_{n\geq0} \subset [0,T]$ such that the left-derivatives $\pt^{-}_{2} u(t_{n}, r) \nearrow \infty$. Hence, for $r' \in (r, h)$, the sequence $u(t_{n}, r') \to \infty$, which contradicts that $u(t_{n}, r') \leq u(0, r') < \infty$ for all $n \in \N$. 

Now, it remains to show that $u$ is continuous in time. Assume for a contradiction that the map $t \mapsto u(t, x_{0})$ is not continuous at $t=t_{0}$ for some $x_{0}$. Since $u$ is decreasing in time, $u(\cdot, x_{0})$ has a negative jump at $t_{0}$. Next, we will investigate the cases $u(t_{0}-, x_{0}) > u(t_{0}, x_{0})$ and $u(t_{0}, x_{0}) > u(t_{0}+, x_{0})$ separately.

Suppose $u(t_{0}-, x_{0}) > u(t_{0}, x_{0})$. By Lipschitz continuity in the second variable, there exists $\delta >0$ such that, writing $\mathcal{R} = (t_{0}-\delta, t_{0}) \times (x_{0} - \delta, x_{0}+\delta)$, 
\begin{IEEEeqnarray}{rCl} \label{E:AssCv}
\inf_{(t,x) \in \mathcal{R}} u(t, x) > u(t_{0}, x_{0} + \delta). 
\end{IEEEeqnarray}

Thus $\mathcal{R} \subseteq \mathcal{C}^{f}_{i}$. Let $t \in (t_{0}-\delta, t_{0})$ and $\tau_{\mathcal{R}} := \inf \{ s \geq 0\,:\, (t+s, \hat X^{t,x,\sigma_{i}}_{t+\tau_{\mathcal{R}}}) \notin \mathcal{R} \}$. Then, by the martingality in the continuation region,
\begin{IEEEeqnarray*}{rCl}
u(t,x_0) &=& \tilde{\E} \bigg[ e^{ \int_0^{\tau_{\mathcal R}} \hat X^{t,x_0, \sigma_{i}}_{t+u} - \lambda_{i} \ud u }u(t+\tau_\mathcal R, \hat X^{t,x_0, \sigma_{i}}_{t+\tau_{\mathcal R}}) \\&&+ \int_{0}^{\tau_{\mathcal{R}}} e^{ \int_0^u  \hat X^{t,x_{0}, \sigma_{i}}_{t+s} -\lambda_{i} \ud s }f(t+u, \hat X^{t,x_{0},\sigma_{i}}_{t+u}) \ud u \bigg]\\
&\leq& \tilde{\E} \bigg[ e^{ (t_0-t)(x_{0}+\delta)^{+} }u(t,x_0+\delta)\Ind_{\{t+\tau_\mathcal R <t_0\}} \\
&&+ e^{ (t_0-t)(x_0+\delta)^{+}}u(t_0,x_0+\delta)\Ind_{\{t+\tau_\mathcal R = t_0\}}\\
&&+ \int_{0}^{t_{0}-t} e^{ \int_0^u  \hat X^{t,x_{0}, \sigma_{i}}_{t+s} -\lambda_{i} \ud s }|f(t+u, \hat X^{t,x_{0},\sigma_{i}}_{t+u})| \ud u \bigg]\\
&\leq& e^{(t_0-t)(x_0+\delta)^+}u(t,x_0+\delta) \tilde{\P}(t+\tau_\mathcal R <t_0) +e^{(t_0-t)(x_0+\delta)^+}u(t_0,x_0+\delta)\\
&&+\int_{0}^{t_{0}-t} \tilde{\E} \lt[ e^{ \int_0^u  \hat X^{t,x_{0}, \sigma_{i}}_{t+s} -\lambda_{i} \ud s } |f(t+u, \hat X^{t,x_{0},\sigma_{i}}_{t+u})|\rt] \ud u \\
&\to& u(t_0,x_0+\delta)
\end{IEEEeqnarray*}  
as $t \to t_{0}$, contradicting \eqref{E:AssCv}.

The other case to consider is $u(t_{0}, x_{0}) > u(t_{0}+, x_{0})$; we look into the situation $u(t_{0}, x_{0}) > u(t_{0}+, x_{0})>1$ first. The local Lipschitz continuity in the second variable and the decay in the first variable imply that there exist $\epsilon>0$ and $\delta>0$ such that, writing $\mathcal{R} = (t_{0}, t_{0}+\epsilon] \times [x_{0}-\delta, x_{0}+\delta]$, 
\begin{IEEEeqnarray}{rCl} \label{E:AssCd}
u(t_{0}, x_{0}) > \sup_{(t,x) \in \mathcal{R}} u(t,x)  \geq \inf_{(t,x) \in \mathcal{R}} u(t,x) > 1.
\end{IEEEeqnarray} 
Hence, $\mathcal R \subseteq \mathcal{C}^{f}_{i}$ and writing $\tau_\mathcal R:=\inf\{s\geq 0: (t_{0}+s, \hat X^{t_{0},x_0, \sigma_{i}}_{t_{0}+s})\notin\mathcal R\}$ we have
\begin{IEEEeqnarray*}{rCl}
u(t_{0},x_0) &=& \tilde{\E} \bigg[ e^{ \int_{0}^{\tau_{\mathcal R}} \hat X^{t_{0},x_0 , \sigma_{i}}_{t_{0}+u} - \lambda_{i} \ud u }
u(t_0+\tau_\mathcal R, \hat X^{t_{0},x_0, \sigma_{i}}_{t_0+\tau_{\mathcal R}}) \\
&&+ \int_{0}^{\tau_{\mathcal{R}}} e^{ \int_0^u  \hat X^{t_0,x_0, \sigma_{i}}_{t_0+s} -\lambda_{i} \ud s }f(t_0+u, \hat X^{t_0,x_0,\sigma_{i}}_{t_0+u}) \ud u \bigg] \\
&\leq& \tilde{\E} \lt[ e^{\epsilon(x_0+\delta)^+}u(t_{0},x_0+\delta)\Ind_{\{\tau_\mathcal R <\epsilon\}}\rt]\\
&&+\tilde{\E} \bigg[e^{\epsilon(x_0+\delta)^+}u(t_0+\epsilon,x_0+\delta)\Ind_{\{\tau_\mathcal R = \epsilon\}} \\&&+\int_{0}^{\epsilon} e^{ \int_0^u  \hat X^{t_0,x, \sigma_{i}}_{t_0+s} -\lambda_{i} \ud s }|f(t_0+u, \hat X^{t_0,x_0,\sigma_{i}}_{t_0+u})| \ud u\bigg]\\
&\leq& e^{\epsilon(x_0+\delta)^+}u(t_{0},x_0+\delta) \tilde{\P}(\tau_\mathcal R <\epsilon) +e^{\epsilon(x_0+\delta)^+}u(t_0+\epsilon,x_0+\delta)\\
&&+\int_{0}^{\epsilon} \tilde{\E} \lt[ e^{ \int_0^u  \hat X^{t_0,x_0, \sigma_{i}}_{t_0+s} -\lambda_{i} \ud s }|f(t_0+u, \hat X^{t_0,x_0,\sigma_{i}}_{t_0+u})| \rt] \ud u \\
&\to& u(t_0+,x_0+\delta)
\end{IEEEeqnarray*}
as $\epsilon \searrow 0$, which contradicts \eqref{E:AssCd}.

Lastly, suppose that $u(t_{0}, x_{0}) > u(t_{0}+, x_{0}) = 1$. By Lipschitz continuity in the second variable, there exists $\delta>0$ such that
\begin{IEEEeqnarray}{rCl}
\label{ringo}
\inf_{x\in(x_{0}-\delta, x_{0})}u(t_{0},x)>u(t_0+,x_0)=1.
\end{IEEEeqnarray}
Consequently, $(t_{0}, T]\times (x_{0}-\delta, x_{0}) \subseteq \mathcal{D}^{f}_{i}$. Hence the process $\hat X^{t_{0}, x_{0}-\delta/2, \sigma_{i}}$ hits the stopping region immediately and so $(t_{0}, x_{0}-\delta/2) \in \mathcal{D}^{f}_{i}$, which contradicts $\eqref{ringo}$.
\end{proof}

\begin{proposition}[Optimal stopping boundary] \label{T:OpJf} \item
Let $f: [0,T] \times (l, h) \to \R $ be bounded, decreasing in the first variable as well as increasing and convex in the second variable. Then the following hold.
\begin{enumerate}[(i)]
\item There exists a function $b^{f}_{\sigma_{i}}: [0,T) \to [l,h]$ that is both increasing, right-continuous with left limits, and satisfies
\begin{IEEEeqnarray}{rCl} \label{E:bSep}
\mathcal{C}^{f}_{i} = \{ (t,x) \in [0,T) \times (l, h) \,:\, x > b^{f}_{\sigma_{i}}(t) \}.
\end{IEEEeqnarray}
\item \label{T:JfPDE}
The pair $(J_{i}f, b^{f}_{\sigma_{i}})$ satisfies  the free-boundary problem
\begin{IEEEeqnarray}{rCl} \label{E:FFBP}
\lt\{
\begin{array}{rl}
 \pt_{t}u(t,x) + {\sigma_{i}} \phi(x, \sigma_{i}) \pt_{x} u(t,x) + \frac{1}{2} \phi(x, \sigma_{i})^{2} \pt_{xx} u(t,x) \\+(x-\lambda_{i})u(t,x)+\lambda_{i} f(t,x) = 0, & 
  \text{ if } x > b^{f}_{\sigma_{i}}(t), \\
u(t,x) = 1, & \text{ if } x \leq b^{f}_{\sigma_{i}}(t) \text{ or } t=T.
\end{array}
\rt.
\end{IEEEeqnarray}
\end{enumerate}
\end{proposition}
\begin{proof}
\begin{enumerate}[(i)]
\item
By Proposition \ref{T:Ji} (\ref{T:ic}), there exists a unique function $b^{f}_{\sigma_{i}}$ satisfying \eqref{E:bSep}. Moreover, by Proposition \ref{T:Ji} (\ref{T:decrt}), this boundary $b^{f}_{\sigma_{i}}$ is increasing. Hence, using Proposition \ref{T:reg}, we also obtain that $b^{f}_{\sigma_{i}}$ is right-continuous with left limits.
\item 
The proof follows a well-known standard argument (e.g.~see \cite[Theorem 7.7 in Chapter 2]{KS2}), thus we omit it. 
\end{enumerate}
\end{proof}

\subsection{A sequence of approximating problems}
\label{ss:approx_seq}

Let us define a sequence of stopping times $\{\xi^{t}_{n}\}_{n \geq 0}$ recursively by
\begin{IEEEeqnarray*}{rCl}
\xi^{t}_{0}&:=&0, \\
\xi^{t}_{n} &:=& \inf \{ s > \xi^{t}_{n-1}\,:\, \sigma( t + s) \neq \sigma(t + \xi^{t}_{n-1})  \}, \quad n > 0.
\end{IEEEeqnarray*}
Here $\xi^{t}_{n}$ represents the duration until the $n$-th volatility jump since time $t$.   Furthermore, let us define a sequence of operators $\{ J^{(n)} \}_{n\geq 0}$ by
\begin{IEEEeqnarray*}{rCl} \label{E:J^n}
(J^{(n)} f)(t,x, \sigma_{i}) &:=& \sup_{\tau \in \mathcal{T}_{T-t}} \tilde{\E} \lt[ e^{\int_0^\tau  \hat X^{t,x, \sigma_{i}}_{t+s}  \ud s } \Ind_{\{ \tau < \xi^{t}_{n} \}} + e^{\int_0^{\xi^{t}_{n}}  \hat X^{t,x, \sigma_{i}}_{t+s} \ud s }f(t+\xi^{t}_{n}, \hat X^{t,x, \sigma_{i}}_{t+\xi^{t}_{n}}) \Ind_{\{\tau \geq \xi^{t}_{n}\}} \rt],\\
{} \IEEEyesnumber
\end{IEEEeqnarray*} 
where $f : [0, T] \times (l,r) \to \R$ is bounded.
In particular, note that $J^{(0)} f = f$ and $J^{(1)}f=Jf$. Similarly as for the operator $J$, we define $J^{(n)}_{i}$ by 
\begin{IEEEeqnarray*}{rCl}
J^{(n)}_{i} f := (J^{(n)}f)(\cdot, \cdot, \sigma_{i}).
\end{IEEEeqnarray*}

\begin{proposition} \label{T:Jchain}
Let $n \geq 0$ and $i \in \{0,\ldots, m\}$. Then

\begin{IEEEeqnarray}{rCl} \label{E:Chaining}
 J^{(n+1)}_{i}= J_{i} \lt( \sum_{j\neq i} \frac{\lambda_{ij}}{\lambda_{i}} J^{(n)}_{j} \rt). 
\end{IEEEeqnarray}
\end{proposition}
\begin{proof} 
The proof is by induction. In order to present the argument of the proof while keeping  intricate notation at bay, we will only prove that, for a bounded $f:[0,T]\times (l,h) \to \R$ and $x\in (l,h)$, the identity $(J^{(2)}_{i}f)(t,x)= (J_{i}(\sum_{j\neq i} \frac{\lambda_{ij}}{\lambda_{i}} J_{j}f))(t,x)$ holds.  The induction step $J^{(n+1)}_{i}= J_{i} \lt( \sum_{j\neq i} \frac{\lambda_{ij}}{\lambda_{i}} J^{(n)}_{j} \rt)$ follows a similar argument, though with more abstract notation. Note that without loss of generality, we can assume $t=0$, which we do.

Firstly, we will show $(J_{i}^{(2)}f)(0,x) \leq J_{i}\bigg( \sum_{j \neq i}\frac{\lambda_{ij}}{\lambda_{i}} (J_{j}f)\bigg)(0,x)$ and then the opposite inequality. For $j \in \N$, we will write $\xi_{j}$ instead of $\xi^{0}_{j}$ as well as will use the notation $\eta_{j}:=\xi_{j}-\xi_{j-1}$.  Let $\tau \in \mathcal{T}_{T}$ and consider  
\begin{IEEEeqnarray*}{rCl}
\IEEEeqnarraymulticol{3}{l}{
A(\tau)} \\ 
 &:=&  \tilde{\E} \lt[  e^{\int_0^\tau  \hat X^{0,x,\sigma_i}_{s}  \ud s } \Ind_{\{ \tau < \eta_{1} \}} +  e^{\int_0^\tau  \hat X^{0,x,\sigma_i}_{s}  \ud s } \Ind_{\{ \eta_{1} \leq \tau < \xi_{2} \}} +e^{\int_0^{\xi_{2}}  \hat X^{0,x,\sigma_i}_{s} \ud s }f(\xi_{2}, \hat X^{0,x, \sigma_i}_{\xi_{2}}) \Ind_{\{\tau \geq \xi_{2}\}}  \rt] \\
&=& \tilde{\E} \bigg[ e^{\int_0^\tau  \hat X^{0,x,\sigma_i}_{s}  \ud s } \Ind_{\{ \tau < \eta_{1} \}} + \tilde{\E} \big[  e^{\int_0^\tau  \hat X^{0,x,\sigma_i}_{s}  \ud s } \Ind_{\{ \eta_{1} \leq \tau < \xi_{2} \}} \\
&&+e^{\int_0^{\xi_{2}}  \hat X^{0,x,\sigma_i}_{s} \ud s }f(\xi_{2}, \hat X^{0,x, \sigma_i}_{\xi_{2}}) \Ind_{\{\tau \geq \xi_{2}\}} \, | \, \F^{\hat X^{0,x,\sigma_{i}}, N}_{\eta_{1}} \big] \bigg], \IEEEyesnumber \label{E:J12}
\end{IEEEeqnarray*}
where $\{N_{t}\}_{t\geq0}$ denotes the process counting the volatility jumps.
The inner conditional expectation in \eqref{E:J12} satisfies
\begin{IEEEeqnarray*}{rCl}
\IEEEeqnarraymulticol{3}{l}{
\tilde{\E} \big[  e^{\int_0^\tau  \hat X^{0,x,\sigma_i}_{s}  \ud s } \Ind_{\{ \eta_{1} \leq \tau < \xi_{2} \}} +e^{\int_0^{\xi_{2}}  \hat X^{0,x,\sigma_i}_{s} \ud s }f(\xi_{2}, \hat X^{0,x, \sigma_i}_{\xi_{2}}) \Ind_{\{\tau \geq \xi_{2}\}} \, | \, \F^{\hat X^{0,x,\sigma_{i}}, N}_{\eta_{1}} \big]} \\
&=& e^{\int_0^{\eta_{1}}  \hat X^{0,x,\sigma_i}_{s}  \ud s } \Ind_{\{ \eta_{1} \leq \tau \}}\tilde{\E} \big[  e^{\int_{\eta_{1}}^{\tau}  \hat X^{0,x,\sigma_i}_{s}  \ud s } \Ind_{\{ \tau < \xi_{2} \}} \\&&+e^{\int_{\eta_{1}}^{\xi_{2}}  \hat X^{0,x,\sigma_i}_{s} \ud s }f(\xi_{2}, \hat X^{0,x, \sigma_i}_{\xi_{2}}) \Ind_{\{\tau \geq \xi_{2}\}} \, | \, \F^{\hat X^{0,x,\sigma_{i}}, N}_{\eta_{1}} \big] \\
&=& e^{\int_0^{\eta_{1}}  \hat X^{0,x,\sigma_i}_{s}  \ud s } \Ind_{\{ \eta_{1} \leq \tau \}} \sum_{j \neq i} \frac{\lambda_{ij}}{\lambda_{i}}
\tilde{\E}^{\eta_{1},\hat X^{0,x,\sigma_i}_{\eta_{1}}, \sigma_{j}} \bigg[ e^{\int_{0}^{\tilde{\tau}} \hat X_{\eta_{1}+s} \ud s} \Ind_{\{\tilde{\tau} < \eta_{2}\}} \\
&&+ e^{\int_{0}^{\eta_{2}} \hat X_{\eta_{1}+s} \ud s }f(\eta_{1}+\eta_{2}, \hat X_{\eta_{1}+\eta_{2}}) \Ind_{\{\tilde{\tau} \geq \eta_{2}\}}\bigg], \IEEEyesnumber \label{E:J12In}
\end{IEEEeqnarray*}
where $\tilde{\tau} = \tau - \eta_{1}$ in the case $\eta_{1} \leq \tau \leq T$. Therefore, substituting \eqref{E:J12In} into \eqref{E:J12} and then taking a supremum over $\tilde{\tau}$, we get
\begin{IEEEeqnarray*}{rCl}
A(\tau) 
&\leq&  
\tilde{\E} \bigg[ e^{\int_0^\tau  \hat X^{0,x,\sigma_i}_{s}  \ud s } \Ind_{\{ \tau < \eta_{1} \}} \\
&&+ e^{\int_0^{\eta_{1}}  \hat X^{0,x,\sigma_i}_{s}  \ud s } \Ind_{\{ \tau \geq \eta_{1} \}}  \sum_{j \neq i} \frac{\lambda_{ij}}{\lambda_{i}} \sup_{\tilde{\tau} \in \mathcal{T}_{T-T\wedge \eta_{1}}}
\tilde{\E}^{\eta_{1},\hat X^{0,x,\sigma_i}_{\eta_{1}}, \sigma_j} \big[ e^{\int_{0}^{\tilde{\tau}} \hat X_{\eta_{1}+s} \ud s} \Ind_{\{\tilde{\tau} < \eta_{2}\}} \\
&&+ e^{\int_{0}^{\eta_{2}} \hat X_{\eta_{1}+s} \ud s }f(\eta_{1}+\eta_{2}, \hat X_{\eta_{1}+\eta_{2}}) \Ind_{\{\tilde{\tau} \geq \eta_{2}\}}\big]\bigg] \\
&=& \tilde{\E} \bigg[ e^{\int_0^\tau  \hat X^{0,x,\sigma_i}_{s}  \ud s } \Ind_{\{ \tau < \eta_{1} \}} + e^{\int_0^{\eta_{1}}  \hat X^{0,x,\sigma_i}_{s}  \ud s } \Ind_{\{ \tau \geq \eta_{1} \}}  \sum_{j \neq i} \frac{\lambda_{ij}}{\lambda_{i}} (J_{j}f)(\eta_{1}, X^{0,x,\sigma_i}_{\eta_{1}})
\bigg] \IEEEyesnumber \label{E:AtauJ}
\end{IEEEeqnarray*}
Taking a supremum over $\tau$ in \eqref{E:AtauJ}, we obtain
\begin{IEEEeqnarray*}{rCl}
(J_{i}^{(2)}f)(0,x) = \sup_{\tau \in \mathcal{T}_{T}} A(\tau) \leq J_{i}\bigg( \sum_{j \neq i}\frac{\lambda_{ij}}{\lambda_{i}} (J_{j}f)\bigg)(0,x). \IEEEyesnumber \label{E:JLEQ}
\end{IEEEeqnarray*}
%

It remains to establish the opposite inequality. Let $\tau \in \mathcal{T}_{T}$ and define
\begin{IEEEeqnarray}{rCl}
\check{\tau} &:=& \tau \Ind_{\{ \tau \leq \eta_{1}  \}} + (\eta_{1} \wedge T +\tau_{\sigma(\eta_{1})}) \Ind_{\{ \tau > \eta_{1} \}}, \label{E:TauCheck}
\end{IEEEeqnarray}
where $\tau_{\sigma(\eta_{1})} := \tau^{f}_{\sigma(\eta_{1})} (\eta_{1} \wedge T, \hat X^{0,x,\sigma_{i}}_{\eta_{1} \wedge T})$. Clearly, $\check{\tau} \in \mathcal{T}_{T}$. Then 
\begin{IEEEeqnarray*}{rCl}
\IEEEeqnarraymulticol{3}{l}{
(J^{(2)}_{1}f)(0,x)} \\ 
&\geq& A(\check{\tau}) \\
&=&  \tilde{\E} \bigg[  e^{\int_0^{\tau}  \hat X^{0,x,\sigma_i}_{s}  \ud s } \Ind_{\{ \tau < \eta_{1} \}} + e^{\int_0^{\eta_{1}}  \hat X^{0,x,\sigma_i}_{s}  \ud s } \Ind_{\{ \tau \geq \eta_{1}\}} \sum_{j \neq i} \frac{\lambda_{ij}}{\lambda_{i}} \tilde{\E}^{\eta_{1},\hat X^{0,x,\sigma_i}_{\eta_{1}}, \sigma_j} \big[ e^{\int_{0}^{\tau_{\sigma_j}} \hat X_{\eta_{1}+s} \ud s} \Ind_{\{\tau_{\sigma_j} < \eta_{2}\}} \\
&&+ e^{\int_{0}^{\eta_{2}} \hat X_{\eta_{1}+s} \ud s }f(\eta_{1}+\eta_{2}, \hat X_{\eta_{1}+\eta_{2}}) \Ind_{\{ \tau_{\sigma_j} \geq \eta_{2}\}}\big] \bigg]  \\
&=& \tilde{\E} \bigg[  e^{\int_0^{\tau}  \hat X^{0,x,\sigma_i}_{s}  \ud s } \Ind_{\{ \tau < \eta_{1} \}} + e^{\int_0^{\eta_{1}}  \hat X^{0,x,\sigma_i}_{s}  \ud s } \Ind_{\{ \tau \geq \eta_{1}\}} \sum_{j \neq i} \frac{\lambda_{ij}}{\lambda_{i}} (J_{j}f)(\eta_{1}, \hat X^{0,x,\sigma_i}_{\eta_{1}})  \bigg],
\end{IEEEeqnarray*}
where Proposition \ref{T:OSTJ} was used to obtain the last equality. Hence, by taking supremum over stopping times $\tau \in \mathcal{T}_{T}$, we get  
\begin{IEEEeqnarray}{rCl}
(J^{(2)}_{i}f)(0,x) \geq J_{i} \bigg( \sum_{j \neq i}\frac{\lambda_{ij}}{\lambda_{i}} (J_{j}f) \bigg)(0,x). \label{E:JGEQ}
\end{IEEEeqnarray}
Finally, \eqref{E:JLEQ} and \eqref{E:JGEQ} taken together imply 
\begin{IEEEeqnarray*}{rCl}
(J^{(2)}_{i}f)(0,x) = J_{i} \bigg( \sum_{j \neq i}\frac{\lambda_{ij}}{\lambda_{i}} (J_{j}f) \bigg)(0,x). 
\end{IEEEeqnarray*}

\end{proof}
\begin{remark}
In \cite{LW}, the authors use the same approximation procedure for an optimal stopping problem with regime switching volatility as in this article. Unfortunately, a mistake is made in equation (18) of \cite{LW}, which wrecks the subsequent approximation procedure when the number of volatility states is greater than $2$. The identity (18) therein should be replaced by \eqref{E:Chaining}.
\end{remark}
\subsection{Convergence to the value function}

\begin{proposition}[Properties of the approximating sequence] \label{T:vlim} \item
\begin{enumerate}[(i)]
\item 
The sequence of functions $\{ J^{(n)} 1 \}_{n \geq 0}$ is increasing, bounded from below by $1$ and from above by $e^{hT}$. 
\item \label{T:Jndic}
Every $J^{(n)} 1 $ is decreasing in the first variable $t$ as well as increasing and convex in the second variable $x$.
\item \label{T:vappr} The sequence of functions 
\begin{IEEEeqnarray*}{rCl} 
J^{(n)} 1 \nearrow v \quad \text{ pointwise as } n \nearrow \infty.
\end{IEEEeqnarray*}
Moreover, the approximation error
\begin{IEEEeqnarray}{rCl} \label{E:ApprErr}
\| v - J^{(n)} 1 \|_{\infty} \leq e^{hT} \lambda T \frac{(\lambda T)^{n-1}}{(n-1)!}   \text{ as } n \to \infty,
\end{IEEEeqnarray}
where $\lambda := \max \{ \lambda_{i}\,:\, 1 \leq i \leq m \}$. 

\item
For every $n \in \N\cup\lt\{0\rt\}$,
\begin{IEEEeqnarray}{rCl} \label{E:Jbounds}
J_{m}^n 1 \leq J^{(n)} 1 \leq J_{1}^n 1.
\end{IEEEeqnarray}
\end{enumerate} 
\end{proposition}
\begin{proof} 
\begin{enumerate}[(i)]
\item The statement that $\{ J^{(n)}_{i} 1 \}_{n \geq 0}$ is increasing, bounded from below by $1$ and from above by $e^{hT}$ is a direct consequence of the definition \eqref{E:J^n}. 
\item
The claim that every $J^{(n)}_{i} 1 $ is decreasing in the first variable $t$ as well as increasing and convex in the second variable $x$ follows by a straightforward induction on $n$, 
%
 using Proposition \ref{T:Ji} (\ref{T:decrt}),(\ref{T:ic}) and Proposition \ref{T:Jchain} at the induction step.
\item First, let $i \in \{ 1, \ldots, m\}$ and note that, for any $n \in \N$,
\begin{IEEEeqnarray*}{rCl}
J^{(n)}_{i}1 \leq v_{i}.
\end{IEEEeqnarray*}
Here the inequality holds by suboptimality, since $J^{(n)}_{i}1$ corresponds to an expected payoff of a particular stopping time in the problem \eqref{E:OSN}. Next, define
\begin{IEEEeqnarray*}{rCl}
U^{(i)}_{n}(t,x) &:=& \sup_{\tau \in \mathcal{T}_{T-t}} \tilde{\E} \lt[  e^{\int_0^\tau  \hat X^{t,x, \sigma_{i}}_{t+s}  \ud s } \Ind_{\{ \tau < \xi^{t}_{n} \}}\rt].
\end{IEEEeqnarray*}
Then 
\begin{IEEEeqnarray}{rCl} \label{E:Sandwich}
U^{(i)}_{n}(t,x) \leq (J^{(n)}_{i}1) (t,x) \leq v_{i}(t,x) \leq U^{(i)}_{n}(t,x) + e^{h(T-t)} \P(\xi^{t}_{n} \leq T-t). 
\end{IEEEeqnarray}
Since it is a standard fact that the $n^{\text{th}}$ jump time, call it $\zeta_{n}$, of a Poisson process with jump intensity $\lambda := \max \{ \lambda_{i}\,:\, 1 \leq i \leq m \}$ follows the Erlang distribution, we have 
\begin{IEEEeqnarray*}{rCl}
\P(\xi^{t}_{n} \leq T-t) &\leq& \P(\zeta_{n} \leq T-t)
\\&=& \frac{1}{(n-1)!} \int_{0}^{\lambda (T-t)} u^{n-1} e^{-u} \ud u \\
&\leq& \lambda T \frac{(\lambda T)^{n-1}}{(n-1)!}. \\
\end{IEEEeqnarray*}

Therefore, by \eqref{E:Sandwich},
\begin{IEEEeqnarray*}{rCl}
\| v - J^{(n)} 1 \|_{\infty}  \leq e^{hT} \lambda T \frac{(\lambda T)^{n-1}}{(n-1)!} \text{ as } n \to \infty.
\end{IEEEeqnarray*}

\item
The string of inequalities \eqref{E:Jbounds} will be proved by induction. First, the base step is obvious. Now, suppose \eqref{E:Jbounds} holds for some $n \geq 0$. Hence, for any $i \in \{1, \ldots, m\}$,
\begin{IEEEeqnarray}{rCl} \label{E:JIH}
J^{n}_{m} 1 \leq  \sum_{j\neq i} \frac{\lambda_{ij}}{\lambda_{i}} J^{(n)}_{j} 1  \leq J^n_1 1.
\end{IEEEeqnarray}
Let us fix $i \in \lt\{ 1, \ldots, m \rt\}$. By Proposition \ref{T:Ji} (\ref{T:ic}), every function in \eqref{E:JIH} is convex in the spatial variable $x$, thus \cite[Theorem 6.1]{ET08} yields
\begin{IEEEeqnarray*}{rCl} 
J^{n+1}_{m} 1 \leq J_i \lt(  \sum_{j\neq i} \frac{\lambda_{ij}}{\lambda_{i}} J^{(n)}_{j} 1 \rt)  \leq J^{n+1}_1 1.
\end{IEEEeqnarray*}
As $i$ was arbitrary, we also have
\begin{IEEEeqnarray}{rCl}  
J_{\sigma_{m}}^{n+1} 1 \leq J^{(n+1)} 1 \leq J_{\sigma_{1}}^{n+1} 1.
\end{IEEEeqnarray}
\end{enumerate}
\end{proof}

\begin{remark}
If instead of $1$ we choose the constant function $e^{hT}$ to apply the operators $J^{(n)}_{i}$ to, then, following the same strategy as above,
$\{ J^{(n)}_{i} e^{hT} \}_{n \geq 0}$ is a decreasing sequence of functions with the limit $ J^{(n)}_{i} e^{hT} \searrow v_{i}$ pointwise as $n \nearrow \infty$.
\end{remark}

Let $\mathcal{B}_{b}([0,T]\times (l,h); \R)$ denote the set of bounded  functions from  $[0,T]\times (l,h)$ to $\R$ and define an operator $
\tilde{J} : \mathcal{B}_{b}([0,T]\times (l,h); \R)^{m} \to \mathcal{B}_{b}([0,T]\times (l,h); \R)^{m}$  by
\begin{IEEEeqnarray*}{rCl}
\tilde{J} \lt( \begin{array}{c}
f_{1} \\ \vdots \\ f_{m}
\end{array}\rt) &:=& \lt( \begin{array}{c}
J_{1} ( \sum_{j \neq 1} \frac{\lambda_{1j}}{\lambda_{1}} f_{j} ) \\ \vdots \\ J_{m} (\sum_{j \neq m} \frac{\lambda_{mj}}{\lambda_{m}} f_{j} )\end{array} \rt).
\end{IEEEeqnarray*}

\begin{proposition} \label{T:AJ}
\begin{enumerate}[(i)]
\item Let $f \in \mathcal{B}_{b}([0,T]\times (l,h); \R)^{m}$. Then 
\begin{IEEEeqnarray*}{rCl}
\lim_{n\to \infty} \tilde{J}^{n} f &=& \lt(\begin{array}{c}
v_{1} \\ \vdots \\ v_{m}
\end{array} \rt).
\end{IEEEeqnarray*}
\item \label{T:JFP}
The vector $(v_{1}, \ldots,  v_{m})^{tr}$ of value functions is a fixed point of the operator $\tilde{J}$, i.e.
\begin{IEEEeqnarray}{rCl} \label{E:JFP}
\tilde{J} \lt( \begin{array}{c}
v_{1} \\ \vdots \\ v_{m} 
\end{array} \rt) &=& \lt( \begin{array}{c}
v_{1} \\ \vdots \\ v_{m}  \end{array}\rt).
\end{IEEEeqnarray}

\end{enumerate}
\end{proposition}
\begin{proof}
\begin{enumerate}[(i)]
\item
Observe that the argument in the proof of part (\ref{T:vappr})~of Proposition \ref{T:vlim} also gives that $J^{(n)}_{i} g \to v_{i}$ as $n \to \infty$ for any bounded $g$. Hence to finish the proof it is enough to recall the relation \eqref{E:Chaining} in Proposition \ref{T:Jchain}.
\item
Let $i \in \{1, \ldots, m\}$.  By Proposition \ref{T:Jchain}, 
\begin{IEEEeqnarray}{rCl} \label{E:Jin1}
 J^{(n+1)}_{i}1= J_{i} \lt( \sum_{j\neq i} \frac{\lambda_{ij}}{\lambda_{i}} J^{(n)}_{j} 1\rt). 
\end{IEEEeqnarray}
By Proposition \ref{T:vlim} (\ref{T:vappr}), for every $j \in \{1, \ldots, m\}$, the sequence $J^{(n)}_{j} 1 \nearrow v_{j}$ as $n \nearrow \infty$, so, letting $n \nearrow \infty$ in \eqref{E:Jin1}, the monotone convergence theorem tells us that 
\begin{IEEEeqnarray}{rCl}
 v_{i}= J_{i} \lt( \sum_{j\neq i} \frac{\lambda_{ij}}{\lambda_{i}} v_{j}  \rt). 
\end{IEEEeqnarray}
   
\end{enumerate}
\end{proof}

\section{The value function and the stopping strategy} \label{S:VFSS}
In this section, we show that the value function $v$ has attractive structural properties and identify an optimal strategy for the liquidation problem \eqref{E:OSM}. The first passage time below a boundary, which is an increasing function of time and volatility, is proved to be optimal. Moreover, we provide a method to approximate the optimal stopping boundary by demonstrating that it is a limit of an increasing sequence of stopping boundaries coming from easier auxiliary problems of Section \ref{S:Appr}.
\begin{theorem}[Properties of the value function] \item \label{T:vprops}
\begin{enumerate}[(i)]
\item
\label{T:vdic}
$v$ is decreasing in the first variable $t$ as well as increasing and convex in the second variable $x$.
\item \label{T:vcts} $v_{i}$ is continuous for every $i \in \{1, \ldots, m\}$.
\item \label{T:vbounds} 
\begin{IEEEeqnarray}{rCl} \label{E:vbounds} 
\check{v}_{\sigma_m} \leq v \leq \check{v}_{\sigma_1},
\end{IEEEeqnarray}
where $\check{v}_{\sigma_i}:[0,T]\times (l, h)  \to \R$ denotes the Markovian value function as in \eqref{E:OSM}, but for a price process \eqref{E:Stock} with constant volatility $\sigma_i$.

\end{enumerate}
\end{theorem}
\begin{proof}
\begin{enumerate}[(i)]
\item Since, by Proposition \ref{T:vlim} (\ref{T:Jndic}), every $J^{(n)} 1$ is decreasing in the first variable $t$, increasing and convex in the second variable $x$, these properties are also preserved in the pointwise limit $\lim_{n \to \infty} J^{(n)}1$, which is $v$ by Proposition \ref{T:vlim} (\ref{T:vappr}).
\item Using part (\ref{T:vdic}) above, the claim follows from Proposition \ref{T:AJ} (\ref{T:JFP}), i.e.~from the fact that $(v_{1}, \ldots, v_{m})^{tr}$ is a fixed point of a regularising operator $\tilde{J}$ in the sense of Proposition \ref{T:reg}. 
\item 
Letting $n \to \infty$ in \eqref{E:Jbounds}, Proposition \ref{T:vlim} (\ref{T:vappr}) gives us \eqref{E:vbounds}.
\end{enumerate}
\end{proof}

For the optimal liquidation problem \eqref{E:OSN} with constant volatility $\sigma$, i.e.~in the case $\sigma_{1}= \ldots =\sigma_{m} =\sigma$, it has been shown in $\cite{EV16}$ that an optimal liquidation strategy is characterised by a increasing continuous stopping boundary $\check{b}_{\sigma} :[0,T)  \to [l, 0]$ with $\check{b}_{\sigma}(T-)=0$ such that the stopping time $\check{\tau}_{\sigma}= \inf \{t \geq 0 \,:\, \hat X_{t} \leq \check{b}_{\sigma}(t) \}\wedge T$ is optimal. It turns out that the optimal liquidation strategy within our regime-switching volatility model shares some similarities with the constant volatility case as the next theorem shows.

\begin{theorem}[Optimal liquidation strategy] \item \label{T:boundaries}
\begin{enumerate}[(i)]
\item \label{T:b}
For every $i \in \{1,\ldots, m\}$, there exists $b_{\sigma_{i}}: [0,T) \to [l, 0]$ that is increasing, right-continuous with left limits, satisfies the equality 
$b_{\sigma_{i}}(T-) = 0$ and the identity
\begin{IEEEeqnarray}{rCl} \label{E:Civj}
\mathcal{C}^{u_{i}}_{i} = \{ (t,x) \in [0,T) \times (l, h) \,:\, x > b_{\sigma_{i}}(t) \}, 
\end{IEEEeqnarray}
where $u_{i} := { \sum_{j\neq i} \frac{\lambda_{ij}}{\lambda_{i}} v_{j}}$. Moreover,
\[
\check{b}_{\sigma_{1}} \leq b_{\sigma_{i}} \leq \check{b}_{\sigma_{m}}.
\]
for any $i \in \lt\{ 1, \ldots, m \rt\}$.
\item \label{T:OSb}
The stopping strategy 
\begin{IEEEeqnarray*}{rCl}
\tau^{*} := \inf \{ s \in [0,T-t) \,:\, \hat X^{t,x, \sigma}_{t+s} \leq b_{\sigma(t+s)}(t+s)  \} \wedge (T-t)\,.
\end{IEEEeqnarray*}
is optimal for the optimal selling problem \eqref{E:OSM}.

\item \label{T:bconv}
For $i \in \{1,\ldots, m\}$, the boundaries 
\begin{IEEEeqnarray*}{rCl}
b_{\sigma_{i}}^{g_i^{(n)}} \searrow b_{\sigma_{i}} \quad \text{pointwise as } n \nearrow \infty,
\end{IEEEeqnarray*}
where $g_i^{(n)} := \sum_{j\neq i} \frac{\lambda_{ij}}{\lambda_{i}}J^{(n)}_{j}1$.
\item 
The pairs $(v_{1}, b_{\sigma_{1}}), (v_{2},b_{\sigma_{2}}), \ldots,(v_{m},b_{\sigma_{m}})$ satisfy  a coupled system of $m$ free-boundary problems with each being
\begin{IEEEeqnarray}{rCl} \label{E:SFBP}
\lt\{
\begin{array}{rl}
 \pt_{t}v_{i}(t,x) + {\sigma_{i}} \phi(x, \sigma_{i}) \pt_{x} v_{i}(t,x) + \frac{1}{2} \phi(x, \sigma_{i})^{2} \pt_{xx} v_{i}(t,x) \\
 +(x- \lambda_{i})v_{i}(t,x)+\sum_{j \neq i} \lambda_{ij}v_{j}(t,x) = 0, & 
  \text{ if } x > b_{i}(t),\\
v_{i}(t,x) = 1, & \text{ if } x \leq b_{i}(t) \text{ or } t = T,
\end{array}
\rt.
\end{IEEEeqnarray}
where $i \in \{1, \ldots, m\}$.
\end{enumerate}
\end{theorem}
\begin{proof}
\begin{enumerate}[(i)]
\item 
The existence of $b_{\sigma_{i}} : [0, T) \to [l, h]$ that is increasing, right-continuous with left limits, and satisfies \eqref{E:Civj} follows from the fixed-point property \eqref{E:JFP}, and Theorem \ref{T:vprops} (\ref{T:vdic}),(\ref{T:vcts}). Since the range of $\check{b}_{\sigma_{1}}, \check{b}_{\sigma_{m}}$ is $[l, 0]$ and $\check{b}_{\sigma_{1}}(T-)= \check{b}_{\sigma_{m}}(T-)=0$, using Theorem \ref{T:vprops} (\ref{T:vbounds}), we also conclude that $\check{b}_{\sigma_{1}} \leq b_{\sigma_{i}}  \leq \check{b}_{\sigma_{m}}$ and that $b_{\sigma_{i}}(T-) = 0$ for every $i$.
\item 
Let us define $\mathcal{D}:= \{ (t,x, \sigma) \in [0,T]\times (l,h)\times \{\sigma_1, \ldots,\sigma_m\} \,:\, v(t,x, \sigma) = 1  \}$. Then $\tau_{\mathcal{D}} := \inf \{ s \geq 0 \,:\, (t+s, \hat X^{t,x, \sigma(t)}_{t+s}, \sigma(t+s)) \in \mathcal{D}\}$ is optimal for the problem \eqref{E:OSM} by \cite[Corollary 2.9]{PS06}. Lastly, from the fixed-point property \eqref{E:JFP} and Proposition \ref{T:OSTJ}, we conclude that $\tau^*=\tau_{\mathcal{D}}$, which finishes the proof.

\item  
Since $J^{(n)}_{i} 1 \nearrow v_{i}$ as $n \nearrow \infty$ and $J^{(n)}_{i} 1 \geq 1$ for all $n$, we have that $\lim_{n\nearrow \infty} b_{\sigma_{i}}^{g_i^{(n)}} \geq b_{\sigma_{i}}$. Also, if $x < \lim_{n\nearrow \infty} b_{\sigma_{i}}^{g_i^{(n)}} (t)$, then $J^{(n)}_{i} 1 (t,x) = 1$ for all $ n \in \N$ and so $v_{i}(t,x)= \lim_{n \nearrow \infty} J^{(n)}_{i} 1 (t,x)=1$. Hence, $\lim_{n\nearrow \infty} b_{\sigma_{i}}^{g_i^{(n)}} \leq b_{\sigma_{i}}$. As a result, $ \lim_{n\nearrow \infty} b_{\sigma_{i}}^{g_i^{(n)}} = b_{\sigma_{i}}$.

\item The free-boundary problem is a consequence of Proposition \ref{T:OpJf} (\ref{T:JfPDE}) and the fixed-point property \eqref{E:JFP}. 
\end{enumerate}
\end{proof}

\begin{remark}
Establishing uniqueness of a classical solution to a time non-homogeneous free-boundary problem is typically a technical task (see \cite{aP} for an example). Not being central to the mission of the paper, the uniqueness of solution to the free-boundary problems \eqref{E:SFBP} and \eqref{E:FFBP} has not been pursued. \end{remark}

\begin{remark}[A possible alternative approach]
It is worth pointing out that a potential alternative approach for the study of the value function and the optimal strategy is to directly analyse the variational inequality formulation (e.g., see~\cite[Section 5.2]{hP}) arising from the optimal stopping problem \eqref{E:OSM}. The coupled system of variational inequalities would need to be studied using weak solution techniques from the PDE theory (e.g., see \cite{BL, hP}) to obtain desired regularity and structural properties of the value function and the stopping region. Though the author is unaware of any work studying exactly this type of free-boundary problem directly in detail, there are available theoretical results \cite{C11} that include existence, uniqueness of viscosity solutions, and a comparison principle for the pricing of American options in regime-switching models. Also, under some conditions, convergence of stable, monotone, and consistent approximation schemes to the value function is shown. Suitable numerical PDE methods and their pros and cons for such a coupled system are discussed in \cite{H11}. 
With this alternative route in mind (provided all the needed technical results can be established), our approach has clear benefits: avoiding many analytical complications that arise in the study of the full system (compare \cite{C11}) and yielding a very intuitive monotone approximation scheme for the value function and the stopping boundary. 
\end{remark}

For further study of the problem in this section, we will make a structural assumption about the Markov chain modelling the volatility.
\begin{assumption} \label{A:Skip-free}
The Markov chain $\sigma$ is \emph{skip-free}, i.e.~for all $i \in \{1, \ldots, m \}$, 
\begin{IEEEeqnarray*}{rCl}
\lambda_{ij}=0 \; \text{ if } j \notin \{i-1, i, i+1 \}.
\end{IEEEeqnarray*}
\end{assumption}
As many popular financial stochastic volatility models have continuous trajectories, and a skip-free Markov chain is a natural discrete state-space approximation of a continuous process, Assumption \ref{A:Skip-free} does not appear to be a severe restriction.

\begin{lemma} \label{L:Eur}
Let $\delta>0$, $g:(l,h)\times [0, \infty) \to [0,\infty) $ be increasing and convex in the first variable as well as decreasing in the second. Then $u : (l, h)\times \{ \sigma_{1}, \ldots, \sigma_{m}\} \to \R$ defined by
\begin{IEEEeqnarray}{rCl} \label{E:EurU}
u(x, \sigma_{i}) := \E \lt[ e^{\int_{0}^{\delta} \hat X^{x, \sigma_{i}}_{u} \ud u} g(\hat X^{x, \sigma}_{\delta}, \sigma(\delta))  \rt]
\end{IEEEeqnarray}
is increasing and convex in the first variable as well as decreasing in the second. 
\end{lemma}
\begin{proof}
We will prove the claim using a coupling argument. Let $(\Omega', \F', \tilde{\P}')$ be a probability triplet supporting a Brownian motion $B$, and two volatility processes $\sigma^{1}$, $\sigma^{2}$ with the state space and transition densities as in \eqref{E:Stock}. In addition, we assume that $B$ is independent of $(\sigma^{1}, \sigma^{2})$, that the starting values satisfy  $\sigma^{1}(0) = \sigma_{i} \leq \sigma_{j} = \sigma^{2}(0)$, and that $\sigma^{1}(t)\leq \sigma^{2}(t)$ for all $t \geq 0$. Also, let $\hat X^{1}$ and $\hat X^{2}$ denote the solutions to \eqref{E:Xhat} when $\sigma$ is replaced by $\sigma^{1}$ and $\sigma^{2}$, respectively.

Let us fix an arbitrary $\omega_{0} \in \Omega'$. Since $\hat W$ is independent of $\sigma^{1}$,
\begin{IEEEeqnarray}{rCl} 
\tilde{E}' \lt[ e^{\int_{0}^{\delta} (\hat X^{1})^{x}_{u} \ud u} g((\hat X^{1})^{x}_{\delta}, \sigma^{1}(\delta))  \,|\, \F^{\sigma^{1}}_{\delta} \rt](\omega_{0})
&=& \tilde{\E}' \lt[ e^{\int_{0}^{\delta} (\tilde X^{1})^{x}_{u} \ud u} g((\tilde X^{1})^{x}_{\delta}, \sigma^{1}(\delta, \omega_{0}))\rt], \label{E:CVol}
\end{IEEEeqnarray}
where $\tilde{X}^{1}$ denotes the process $\hat X^{1}$ with the volatility process $\sigma^{1}$ replaced by a deterministic function $\sigma^{1}(\cdot, \omega_{0})$. Furthermore, the right-hand (and so the left-hand side) in \eqref{E:CVol} as a function of $x$ is increasing by \cite[Theorem~IX.3.7]{RY} as well as convex by  \cite[Theorem 5.1]{ET08}. Hence
\begin{IEEEeqnarray*}{rCl}
u(\cdot, \sigma_{i}) : x \mapsto  \tilde{\E}' \lt[  \tilde{\E}' \lt[ e^{\int_{0}^{\delta} (\hat X^{1})^{x}_{u} \ud u} g((\hat X^{1})^{x}_{\delta}, \sigma^{1}(\delta))  \,|\, \F^{\sigma^{1}}_{\delta} \rt] \rt]
\end{IEEEeqnarray*}
is increasing and convex. Next, we observe that
\begin{IEEEeqnarray}{rCl}
\IEEEeqnarraymulticol{3}{l}{
\tilde{\E}' \lt[ e^{\int_{0}^{\delta} (\hat X^{1})^{x}_{u} \ud u} g((\hat X^{1})^{x}_{\delta}, \sigma^{1}(\delta))  \,|\, \F^{\sigma^{1}, \sigma^{2}}_{\delta} \rt](\omega_{0})} \nonumber \\ 
 \quad &\geq&
\tilde{\E}' \lt[ e^{\int_{0}^{\delta}  (\hat X^{2})^{\delta}_{u} \ud u} g((\hat X^{2})^{x}_{\delta}, \sigma^{1}(\delta))  \,|\, \F^{\sigma^{1}, \sigma^{2} }_{\delta} \rt](\omega_{0}) \nonumber \\
&\geq& \tilde{\E}' \lt[ e^{\int_{0}^{\delta}  (\hat X^{2})^{x}_{u} \ud u} g( (\hat X^{2})^{x}_{\delta}, \sigma^{2}(\delta)) \,|\, \F^{\sigma^{1}, \sigma^{2}}_{\delta} \rt](\omega_{0}). \label{E:CEIneq}
\end{IEEEeqnarray}
In the above, having in mind that the conditional expectations can be rewritten as ordinary expectations similarly as in \eqref{E:CVol},  the first inequality followed by \cite[Theorem 6.1]{ET08}, the second by the decay of $g$ in the second variable. Integrating both sides of \eqref{E:CEIneq} over all possible $\omega_{0} \in \Omega'$ with respect to $\vd \P'$, we get that
\begin{IEEEeqnarray*}{rCl}
u(x, \sigma_{1}) \geq u(x, \sigma_{2}).
\end{IEEEeqnarray*}
Thus we can conclude that $u$ is increasing and convex  in the first variable as well as decreasing in the second. 
\end{proof}

\begin{theorem}[Ordering in volatility] \label{T:OrdVol}
\item \label{T:VIC}
\begin{enumerate}[(i)]
\item \label{T:vineq} $v$ is decreasing in the volatility variable, i.e. 
\begin{IEEEeqnarray*}{rCl}
v_{\sigma_{1}} \geq v_{\sigma_{2}} \geq \ldots \geq v_{\sigma_{m}} .
\end{IEEEeqnarray*} 
\item
The boundaries are ordered in volatility as
\[
b_{\sigma_{1}} \leq b_{\sigma_{2}} \leq \ldots \leq b_{\sigma_{m}}.
\]
\end{enumerate}
\end{theorem}

\begin{proof}
\begin{enumerate}[(i)]

\item We will prove the claim by approximating the value function $v$ by a sequence of value functions $\{v_{n}\}_{n\geq 0}$ of corresponding Bermudan optimal stopping problems. Let $v_{n}$ denote the value function as in \eqref{E:OSM}, but when stopping is allowed only at times $\lt\{ \frac{kT}{2^{n}} \,:\, k\in \{0,1, \ldots, 2^{n}\}\rt\}$.

Let us fix $n \in \N$. We will show that, for any given $k\in \{0, \ldots,2^n\}$ and any $t \in [\frac{k}{2^n}T, T]$, the value function $v_{n}(t,x, \sigma)$ is increasing and convex in $x$ as well as decreasing in $\sigma$ (note that here $\sigma$ denotes the initial value of the process $t\mapsto \sigma(t)$). The proof is by backwards induction from $k=2^n$ down to $k=0$. Since $v_{n}(T, \cdot, \cdot)=1$, the base step $k=2^n$ holds trivially. Now, suppose that, for some given $k \in \{0, \ldots, 2^n\}$, the value $v_{n}(t,x, \sigma)$ is increasing and convex in $x$ as well as decreasing in $\sigma$ for any $t\in [\frac{k}{2^n}T, T]$. Then, Lemma \ref{L:Eur} tells us that for any fixed $t \in [\frac{(k-1)T}{2^n}, \frac{kT}{2^n})$,
\begin{IEEEeqnarray*}{rCl}
f(t, x, \sigma) :=  \tilde{\E} \lt[ e^{\int_{t}^{\frac{kT}{2^n}} \hat X^{t,x, \sigma}_{u} \ud u} v_{n}\lt(\frac{kT}{2^n}, \hat X^{t, x, \sigma}_{\frac{kT}{2^n}}, \sigma \lt(\frac{kT}{2^n} \rt) \rt)  \rt],
\end{IEEEeqnarray*} 
is increasing and convex in $x$ as well as decreasing in $\sigma$. Consequently, since
\begin{IEEEeqnarray}{rCl} \label{E:wStep}
v_{n}(t,x, \sigma) &=& \lt\{ \begin{array}{ll}
f(t,x, \sigma), & t \in (\frac{(k-1)T}{2^n}, \frac{kT}{2^n}), \\
f(t,x, \sigma) \vee 1, & t = \frac{(k-1)T}{2^n}, 
\end{array} \rt.
\end{IEEEeqnarray}
the value $v_{n}(t,x, \sigma)$ is increasing and convex in $x$ as well as decreasing in $\sigma$ for any fixed $t \in [ \frac{k-1}{2^n}T, T]$. Hence, by backwards induction, $v_{n}$ is increasing and convex in the second argument $x$ as well as decreasing in the third argument $\sigma$. 

Finally, since $v_{n} \to v$
pointwise as $n \to \infty$, we can conclude that the value function $v$ is decreasing in $\sigma$.

\item From the proof of Theorem \ref{T:boundaries} (\ref{T:OSb}), the claim is a direct consequence of part (\ref{T:vineq}) above.
 
\end{enumerate}
\end{proof}

\begin{remark} \item
\begin{enumerate}
\item
The value function is decreasing in the initial volatility (Theorem \ref{T:OrdVol} (\ref{T:vineq})) also when the volatility is any continuous time-homogeneous positive Markov process independent of the driving Brownian motion $W$. The assertion is justified by inspection of the proof of Lemma \ref{L:Eur} in which no crossing of the volatility trajectories was important, not the Markov chain structure. 

\item 
Though there are no grounds to believe that any of the boundaries $b_{\sigma_{1}}, \ldots ,b_{\sigma_{m}}$ is discontinuous, proving their continuity, except for the lowest one, is beyond the power of customary techniques. Continuity of the lowest boundary can be proved similarly as in the proof of part 4 of \cite[Theorem 3.10]{EV16}, exploiting the ordering of the boundaries. The stumbling block for proving continuity of the upper boundaries is that, at a downward volatility jump time, the value function has a positive jump whose magnitude is difficult to quantify.
\end{enumerate}
\end{remark}

\section{Generalisation to an arbitrary prior}

In this section, we generalise most results of the earlier parts to the general prior case. In what follows, the prior $\mu$ of the drift is no longer a two-point but an arbitrary probability distribution. 


%
%

\subsection{Two-dimensional characterisation of the posterior distribution}

\label{ss:2D-post}
Let us first think a bit more abstractly to develop intuition for the arbitrary prior case. According to the Kushner-Stratonivich stochastic partial differential equation (SPDE) for the posterior distribution (see Section 3.2 of \cite{CR}), if we take the innovation process driving the SPDE and the volatility as the available information sources, then the posterior distribution is a measure-valued Markov process. Unfortunately, there does not exist any applicable general methods to solve optimal stopping problems for measure-valued stochastic processes. If only we were able to characterise the posterior distribution process by an $\R^n$-valued Markovian process (with respect to the filtration generated by the innovation and the volatility processes), then we should manage to reduce our optimal stopping problem with a stochastic measure-valued underlying to an optimal stopping problem with a $\R^n$-valued Markovian underlying. Mercifully, this wishful thinking turns out to be possible in reality as we shall soon see.
 
Unlike in the problem with constant volatility studied in \cite{EV16}, when the volatility is varying, the pair consisting of the elapsed time $t$ and the posterior mean $\hat X_{t}$ is not sufficient (with an exception of the two-point prior case studied before) to characterise the posterior distribution $\mu_{t}$ of $X$ given $\F^{S, \sigma}_{t}$. Hence we need some additional information to describe the posterior distribution. Quite surprisingly, all this needed additional information can be captured in a single additional observable statistic which we will name the `effective learning time'. We start the development by first introducing some useful notation.

Define $Y^{(i)}_t:=Xt + \sigma_{i} W_{t}$ and let $\mu^{(i)}_{t,y}$ denote the posterior distribution of $X$ at time $t$ given $Y^{(i)}_{t}=y$. It needs to be mentioned that, for any given prior $\mu$, the distributions of $X$ given $\F^{Y^{(i)}}_{t}$ and $X$ given $Y^{(i)}_{t}$ are equal (see Proposition 3.1~in \cite{EV16}), which justifies our conditioning only on the last value $Y^{(i)}_{t}$. Also, recall that $l= \inf \supp (\mu)$, $h = \sup \supp (\mu)$.

The next lemma provides the key insight allowing to characterise the posterior distribution by only two parameters.
\begin{lemma} \label{T:EqPost}
Let $\sigma_{2} \geq \sigma_{1} > 0$. Then  
\[
\{\mu^{(1)}_{t,y} \,:\, t  > 0, \, y \in \R \}= \{ \mu^{(2)}_{t,y}\,:\, t  > 0, \, y \in \R  \},
\]
i.e.~the sets of possible conditional distributions of $X$ in both cases are the same.
\end{lemma}
\begin{proof}
Let $t>0$, $y \in \R$. By the standard filtering theory (a generalised Bayes' rule),
\begin{IEEEeqnarray}{rCl}
\label{E:mu}
\mu^{(i)}_{t,y}(\vd u):=\frac{e^{\frac{2uy-u^2t}{2\sigma_{i}^2}} \mu(\vd u)}{\int_\R e^{\frac{2uy-u^2t}{2\sigma_{i}^2}} \mu(\vd u)}.
\end{IEEEeqnarray}
Then taking $r = \lt( \frac{\sigma_{1}}{\sigma_{2}} \rt)^{2}t$ and $y_{1} = \lt( \frac{\sigma_{1}}{\sigma_{2}} \rt)^{2} y$, we have that
\begin{IEEEeqnarray*}{rCl}
\mu^{(2)}_{t,y}(\vd u) &=& \mu^{(1)}_{r ,y_{1}}(\vd u).
\end{IEEEeqnarray*}
\end{proof}
From Lemma \ref{T:EqPost} and \cite[Lemma 3.3]{EV16} we obtain the following important corollary, telling us that, having fixed a prior, any possible posterior distribution can be fully characterised by only two parameters. 

\begin{corollary} \label{T:EqDistr}
Let $t >0$. Then, for any posterior distribution $\mu_{t}(\cdot) = \P(X \in \cdot \,|\, \F^{S, \sigma}_{t})(\omega)$, there exists $(r, x) \in (0,T]\times (l,h)$ such that $\mu_{t}= \mu^{(1)}_{r,y_{1}(r,x)}$, where $y_{1}(r,x)$ is defined as the unique value satisfying $\E[X \,|\, Y^{(1)}_{r}=y_{1}(r,x) ] =x$. In particular, we can take $r = \int_0^t \lt( \frac{\sigma_1}{\sigma(u)(\omega)}\rt)^2 \ud u$ and $y_1(r,x) = \int_0^t \lt( \frac{\sigma_1}{\sigma(u)(\omega)}\rt)^2 \ud Y_u(\omega)$, where $Y_u = \log(S_u) + \frac{1}{2}\int_0^u \sigma(b)^2 \ud b$.
\end{corollary}

When the volatility varies, so does the speed of learning about the drift. The corollary tells us that we can interpret $r$ as the \emph{effective learning time} measured under the constant volatility $\sigma_{1}$. The intuition for the name is that even though the volatility is varying over time, the same posterior distribution $\mu_t$ can be also be obtained in a constant volatility model with the constant volatility $\sigma_1$, just at a different time $r$ and at a different value of the price $S$.

\begin{remark}
It is worth remarking that Corollary \ref{T:EqDistr} also holds for any reasonable positive volatility process. Indeed, using the Kallianpur-Striebel formula with time-dependent volatility (see Theorem 2.9 on page 39 of \cite{CR}), the proof of Lemma \ref{T:EqPost} equally applies for an arbitrary positive time-dependent volatility and immediately yields the result of the corollary. \end{remark}

Next, we make a convenient technical assumption about the prior distribution $\mu$.\begin{assumption}
\label{assump}
The prior distribution $\mu$ is such that
\begin{enumerate}
\item
$
 \int_\R e^{a u^2}\mu(\vd u)<\infty 
$ for some $a>0$,
\item $\psi(\cdot,\cdot) :[0,T]\times (l,h) \to \R$ defined by
\[ \psi(t,x) := \frac{1}{\sigma_{1}} \lt( \E[X^{2}\,|\, Y^{1}_{t}=y_{1}(t,x)] - x^{2} \rt) = \frac{1}{\sigma_{1}} \Var \lt( X\,|\, Y^{1}_{t}= y_{1}(t,x) \rt) \]
is a bounded function that is Lipschitz continuous in the second variable.
\end{enumerate}
\end{assumption}
In particular, all compactly supported distributions as well as the normal distribution are known to satisfy Assumption \ref{assump} (see \cite{EV16}), so it is an inconsequential restriction for practical applications. 

\subsection{Markovian embedding}
Similarly as in the two-point prior case, we will study the optimal stopping problem \eqref{E:OSN2} by embedding it into a Markovian framework. With Corollary \ref{T:EqDistr} telling us that the effective learning time $r$ and the posterior mean $x$  fully characterise the posterior distribution, now, we can embed the optimal stopping problem \eqref{E:OSN2} into the standard Markovian framework by defining the Markovian value function
\begin{IEEEeqnarray*}{rCl} \label{E:GMVF}
v(t,x, r, \sigma) &:=& \sup_{\tau \in \mathcal{T}_{T-t}} \tilde{\E} \lt[ e^{ \int_0^\tau  \hat X^{t, x, r, \sigma}_{t+s} \ud s }\rt], \, (t, x,r, \sigma) \in [0,T] \times (l, h) \times [0,T] \times \{ \sigma_{1},\ldots, \sigma_{m} \}. \\
{}  \label{E:MV} \IEEEyesnumber 
\end{IEEEeqnarray*}

Here the process $\hat X= \hat X^{t,x, r, \sigma_{i}}$ evolves according to
\begin{IEEEeqnarray}{rCl} \label{E:Bhm}
\lt\{ \begin{array}{ll}
\vd \hat X_{t+s} = \sigma_{1} \psi(r_{t+s}, \hat X_{t+s}) \ud s + \frac{\sigma_{1}}{\sigma(t+s)}\psi(r_{t+s}, \hat X_{t+s}) \ud B_{t+s}, & \quad s \geq 0, \\
\vd  r_{t+s} = \lt(\frac{\sigma_{1}}{\sigma(t+s)}\rt)^{2} \vd s, & \quad s \geq 0, \\
\hat X_t = x,\\
r_t = r, \\
\sigma(t)= \sigma_{i};
\end{array} \rt.
\end{IEEEeqnarray} 
the given dynamics of $\hat X$ is a consequence of Corollary \ref{T:EqDistr} and the evolution equation of $\hat X$ in the constant volatility case (see the equation (3.9) in \cite{EV16}).
Also, in \eqref{E:Bhm}, the process $B_{t}= \int_{0}^{t} \sigma(u) \ud u + \hat W_{t}$ is a $\tilde{\P}$-Brownian motion. Lastly, in \eqref{E:MV},
$\mathcal T_{T-t}$ denotes the set of stopping times less than or equal to $T-t$ with respect to the usual augmentation of the filtration generated by $\{ \hat X^{t, x, r, \sigma_{i}}_{t+s}\}_{s \geq 0}$ and $\{ \sigma(t+s)\}_{s \geq 0}$. 

\begin{remark}
Let us note that in light of the observations of Section \ref{ss:2D-post}, if the regime-switching volatility was replaced by a different stochastic volatility process,  the same Markovian embedding \eqref{E:GMVF} could still be useful for the study of the altered problem.    
\end{remark}

\subsection{Outline of the approximation procedure and main results}

Under an arbitrary prior, the approximation procedure of Section \ref{S:Appr} can also be applied, however, the operators $J$ and $J^{(n)}$ need to be redefined in a suitable way. We redefine the operator $J$ to act on a function $f:[0, T] \times (l,h) \times [0,T] \to \R$ as 
\begin{IEEEeqnarray}{rCl} 
\IEEEeqnarraymulticol{3}{l}{
(J f)(t, x, r, \sigma_{i})
} \nonumber \\
  &:=& \sup_{\tau \in \mathcal{T}_{T-t}} \tilde{\E} \lt[ e^{ \int_0^\tau  \hat X^{t, x, r, \sigma_{i}}_{t+s} \ud s } \Ind_{\{\tau < \eta^{t}_{i} \}} + e^{ \int_0^{\eta^{t}_{i}}  \hat X^{t, x, r, \sigma_{i}}_{t+s}  \ud s }f\lt(t+\eta^{t}_{i}, \hat X^{t, x, r, \sigma_{i}}_{t+\eta^{t}_{i}}, r^{t,r}_{t+\eta^{t}_{i}}\rt) \Ind_{\{\tau \geq \eta^{t}_{i}\}}  \rt] \nonumber \\
&=& \sup_{\tau \in \mathcal{T}_{T-t}} \tilde{\E} \lt[ e^{ \int_0^\tau  \hat X^{t, x, r, \sigma_{i}}_{t+s} -\lambda_{i} \ud s } + \int_{0}^{\tau} e^{ \int_0^u  \hat X^{t, x, r, \sigma_{i}}_{t+s} -\lambda_{i} \ud s }f \lt(t+u, \hat X^{t, x, r, \sigma_{i}}_{t+u}, r^{t,r}_{t+\eta^{t}_{i}} \rt) \ud u \rt]  \label{E:J_iGen}
\end{IEEEeqnarray}
and then the operator $J_{i}$ as $J_{i} f := (Jf)(\cdot, \cdot, \sigma_{i})$. Intuitively, $(J_{i} f)$ represents a Markovian value function corresponding to optimal stopping before $t+ \eta^{t}_{i}$, i.e.~before the first volatility change after $t$, when, at time $t+\eta^{t}_{i} < T$, the payoff $f\lt(t+\eta^{t}_{i}, \hat X^{{t, x, r, \sigma_{i}}}_{t+\eta^{t}_{i}}, r^{t,r}_{t+\eta^{t}_{i}} \rt)$ is received, provided stopping has not occurred yet. The underlying process in the optimal stopping problem $J_{i} f$ is the diffusion $(t, \hat X_{t}, r_{t})$. 

The majority of the results in Sections \ref{S:Appr} and \ref{S:VFSS} generalise nicely  to an arbitrary prior case. Proposition \ref{T:Ji} extends word by word; the proofs are analogous, just the second property of $\psi$ from \cite[Proposition 3.6]{EV16} needs to be used for Proposition \ref{T:Ji} (\ref{T:ic}). In addition,  we have that
$f$ decreasing in $r$ implies that $J_{i}f$ is decreasing in $r$, which is proved by a Bermudan approximation argument as in Proposition \ref{T:Ji} (\ref{T:ic}) using the time decay of $\psi$ from \cite[Proposition 3.6]{EV16}. As a result, for $f :[0,T] \times  (l,h) \times [0,T]  \to \R$ that is decreasing in the first and third variables as well as increasing (though not too fast as $x \nearrow \infty$) and convex in the second, there exists a function (a stopping boundary) $b^{f}_{\sigma_{i}} : [0,T)\times [0, T)  \to [l,0]$ that is increasing in both variables and such that the continuation region $  \mathcal{C}^{f}_{i} := \{  (t, x, r) \in [0,T) \times (l,h) \times \big[0, T \big)  \,:\, (J_{i}f) (t,x, r) > 1 \}$ (optimality shown as in Proposition \ref{T:OSTJ}) satisfies
\begin{IEEEeqnarray*}{rCl}
 \mathcal{C}^{f}_{i} &=& \{ (t, x, r) \in [0,T)  \times (l,h) \times \lt[0, T  \rt) \,:\, x > b^{f}_{\sigma_{i}}(t,r) \}.
\end{IEEEeqnarray*}
In addition, each pair  $(J_{i}f,b_{\sigma_{i}}^{f})$ solves the free-boundary problem
\begin{IEEEeqnarray*}{rCl} 
\lt\{
\begin{array}{ll}
 \pt_{t}u(t,x,r) + \lt(\frac{\sigma_{1}}{\sigma_{i}} \rt)^{2} \pt_{r} u(t,x,r) + {\sigma_{1}} \psi(r, x) \pt_{x} u(t,x,r) \\+ \frac{1}{2} \lt(\frac{\sigma_{1}}{\sigma_{i}}\rt)^{2}\psi(r, x)^{2} \pt_{xx} u(t,x,r) +(x-\lambda_{i})u(t,x,r)+\lambda_{i} f(t,x,r) = 0,  \text{ if }x > b_{\sigma_{i}}^{f}(t,r), &\\
u(t,x,r) = 1, \text{ if } x \leq b_{\sigma_{i}}^{f}(t,r) \text{ or } t = T. &
\end{array}
\rt.
\end{IEEEeqnarray*}

With the operator $J^{(n)}$ redefined as
\begin{IEEEeqnarray*}{rCl}
(J^{(n)} f)(t,x, r, \sigma_{i}) &:=& \sup_{\tau \in \mathcal{T}_{T-t}} \tilde{\E} \bigg[ e^{\int_0^\tau  \hat X^{t,x, r, \sigma_i}_{t+s} \ud s } \Ind_{\{ \tau < \xi^{t}_{n} \}} \\
&&+ e^{\int_0^{\xi^{t}_{n}}  \hat X^{t, x, r, \sigma_i}_{t+s} \ud s }f(t+\xi^{t}_{n}, \hat X^{t,x, r, \sigma_i}_{t+\xi^{t}_{n}}, r^{t,r}_{t +\xi^{t}_{n}}) \Ind_{\{\tau \geq \xi^{t}_{n}\}} \bigg],
\end{IEEEeqnarray*}
the crucial Proposition \ref{T:Jchain} holds word by word.
Furthermore, the sequence of functions
 $\{ J^{(n)}
  1 \}_{n \geq 0}$ is increasing, bounded from below by $1$ with each $J^{(n)} 1 $ being decreasing in the first and third variables as well as increasing and convex in the second variable $x$. As desired,
\[
J^{(n)} 1 \nearrow v \quad \text{ pointwise as } n \nearrow \infty,
\]
so the value function $v$ is decreasing in the first and third variables as well as increasing and convex in the second variable; again, $v$ is a fixed point of $\tilde J$. Moreover, the uniform approximation error result \eqref{E:ApprErr} also holds for compactly supported priors (with an obvious reinterpretation $h = \sup (\supp \mu)$). We can also show (by a similar argument as in Theorem \ref{T:boundaries} (\ref{T:bconv})) that
\begin{IEEEeqnarray*}{rCl}
b_{\sigma_{i}}^{g_i^{(n)}} \searrow b_{\sigma_{i}} \quad \text{pointwise as } n \nearrow \infty,
\end{IEEEeqnarray*}
where $g_i^{(n)} := \sum_{j\neq i} \frac{\lambda_{ij}}{\lambda_{i}}J^{(n)}_{j}1$ and the limit $b_{\sigma_{i}}$ is a function increasing in both variables. Lastly, by similar arguments as before, the stopping time
\begin{IEEEeqnarray*}{rCl}
\tau^{*}= \inf \{ s \in [0, T-t) \,:\, \hat X^{t,x,r,\sigma}_{t+s} \leq b_{\sigma(t+s)}(t+s, r_{t+s})  \} \wedge (T-t)\,
\end{IEEEeqnarray*}
is optimal for the liquidation problem \eqref{E:OSN2}.

\begin{remark} The higher volatility, the slower learning about the drift, so under Assumption \ref{A:Skip-free} it is tempting to expect that the value function $v$ is decreasing in the volatility variable and so the stopping boundaries $b_{\sigma_{1}} \leq b_{\sigma_{2}} \leq \ldots \leq b_{\sigma_{m}}$ also in the case of an arbitrary prior distribution $\mu$. Regrettably, proving (or disproving) such monotonicity in volatility has not been achieved by the author.

\end{remark}


\newpage
\bibliographystyle{plain}

\end{document}